\documentclass[sn-mathphys-num]{sn-jnl}

\usepackage{graphicx,amsmath,amssymb,amsfonts,amsthm,physics,mathrsfs} 

\usepackage{mlmodern}

\usepackage{algpseudocode}

\theoremstyle{thmstyleone}%
\newtheorem{theorem}{Theorem}
\newtheorem{corollary}{Corollary}

\newcommand{\sset}[1]{\mathcal{#1}}
\DeclareMathOperator{\sgn}{sgn}
\newcommand{\mdII}{V${}_2$}
\newcommand{\IIind}{{\text{V}_2}}

\begin{document}

\title{Self-contained relaxation-based dynamical Ising machines}

\author{\fnm{Mikhail} \sur{Erementchouk}}\email{merement@gmail.com}
\author{\fnm{Aditya} \sur{Shukla}}\email{aditshuk@umich.edu}
\author{\fnm{Pinaki} \sur{Mazumder}}\email{pinakimazum@gmail.com}
\affil{\orgdiv{Department of Electrical Engineering and Computer Science},
  \orgname{University of Michigan}, \orgaddress{\city{Ann Arbor},
    \postcode{48104}, \state{MI}, \country{USA}}}


\abstract {
  Dynamical Ising machines are based on continuous dynamical systems
  evolving from a generic initial state to a state strongly related to the
  ground state of the classical Ising model on a graph. Reaching the ground
  state is equivalent to finding the maximum (weighted) cut of the graph,
  which presents the Ising machines as an alternative way to solving and
  investigating NP-complete problems. Among the dynamical models,
  relaxation-based models are distinguished by their relations with
  guarantees of performance achieved in time scaling polynomially with the
  problem size. However, the terminal states of such machines are
  essentially non-binary, necessitating special post-processing relying on
  disparate computing. We show that an Ising machine implementing a special
  continuous dynamical system (called the V${}_2$ model) solves the rounding
  problem dynamically. We prove that the V${}_2$ model, starting from an
  arbitrary non-binary state, terminates in a state that trivially rounds
  to a binary state with the cut at least as big as obtained by optimal
  rounding of the initial state. Besides showing that relaxation-based
  dynamical Ising machines can be made self-contained, this result presents
  a non-Boolean realization of solving a non-trivial information processing
  task on Ising machines. Moreover, we prove that if the initial state of
  the V${}_2$-machine is a random limited amplitude perturbation of a binary
  state, the machine progresses to a state with at least as high cut as
  that of the initial binary state. Since the probability of improving the
  cut is finite, this shows that the V${}_2$-machine with random agitations
  converges to a maximum cut state almost surely. }


\maketitle

\section{Introduction}

Computational capabilities of the Ising model, the classical spin system on
a graph, attract researchers' attention for a long
time~\cite{kirkpatrickOptimization1983, hopfieldNeurons1984,
  cernyThermodynamical1985, fuApplication1986}. These capabilities can be
attributed to two key properties of the model ground state, the spin
distribution with the lowest energy. First, the ground state solves a
quadratic unconstrained binary optimization
problem~\cite{kochenbergerBinary2013}. Second, by associating the spin
distribution in the ground state with partition of the nodes of the graph,
one obtains the maximum (weighted) cut of the model
graph~\cite{barahonaComputational1982}. Finding the maximum cut is an
NP-complete problem~\cite{karpReducibility1972,gareySimplified1976}, which
puts the Ising model into the general computing
perspective~\cite{lucasIsing2014}.

The focused effort exploring efficient ways of finding the ground state of
the Ising model has led to the emergence of the class of dynamical Ising
machines. These are essentially continuous dynamical systems that evolve
from a generic initial state to a state tightly related to the Ising model
ground state. Consequently, the defining feature of the Ising machines is
to utilize the ability of special continuous dynamical systems to
effectively minimize the energy of the spin distribution.

While investigating Ising machines, the main attention is paid to dynamical
systems with emergent binary (or, in some sense, close to binary)
states~\cite{aaditMassively2022, tatsumuraLargescale2021,
  tatsumuraScaling2021, patelLogically2022, yamamotoSTATICA2021,
  yamaoka20kSpin2016, ahmedProbabilistic2021,moy968node2022,
  afoakwaBRIM2021, leleuScaling2021}. In these machines, the spins are
represented by continuous dynamical variables, with the coupling energy
usually mimicking the spin coupling in the Ising model. The notable
exception is the class of machines implementing the Kuramoto
model~\cite{kuramotoSelfentrainment1975, shinomotoPhase1986,
  moriDissipative1998, acebronKuramoto2005, albertssonUltrafast2021} of
synchronization in a network of coupled phase
oscillators~\cite{wangSolving2021, wangOIM2019}. In such implementations,
the coupling between the dynamical variables is nonlinear and can be
related to the scalar product of unit vectors representing individual
spins~\cite{shinomotoPhase1986, erementchoukComputational2022}. The
emergence of close-to-binary states in all these machines is forced by the
specially constructed energy landscape for individual spins (see, for
instance,~\cite{bohmOrderofmagnitude2021}).

Despite the evolution of the Ising machines towards the optima of the
objective function, the quality of obtained solutions remains an open
problem. In~\cite{erementchoukComputational2022}, we have demonstrated that
the results obtained within the combinatorial optimization theory can be
applied to analyze the solutions produced by selected machines. We applied
this approach to analyze the computing power of Ising machines based on the
Kuramoto model. Essentially, it stems from the correspondence between the
machine evolution and the gradient descent solution of rank-2 semidefinite
programming (SDP) relaxation \cite{goemans879approximation1994,
  goemansImproved1995} of the max-cut problem, also called
Burer-Monteiro-Zhang (BMZ) heuristic~\cite{burerRankTwo2002}. Using the SDP
relaxation as an underlying computing principle is beneficial because it
can reach the theoretical limit on the performance
guarantee~\cite{raghavendraOptimal2008, khotOptimal2004}. Although this
guarantee was not proven for the BMZ heuristic~\cite{burerLocal2005,
  burerNonlinear2003} (see, however,~\cite{boumalNonconvex2016,
  boumalDeterministic2020, bandeiraLowrank2016}), the practical
implementation of this heuristic (\texttt{Circut}~\cite{burerRankTwo2002})
is among the best solvers of the max-cut problem~\cite{dunningWhat2018}.

In~\cite{shuklaScalable2022, shuklaCustom2023}, we explored an alternative
approach to designing an Ising machine. Instead of forcing the emergence of
close-to-binary states, we focused on the computational capabilities of the
dynamical model driving the machine. In~\cite{shuklaScalable2022}, we
introduced a simplified almost-linear dynamical model governing the Ising
machine based on the BMZ heuristic. We have shown that the introduced model
produces solutions characterized by the integrality gap close to that of
the SDP relaxation.
The simplified model was used in~\cite{shuklaCustom2023} to implement a
custom analog integrated circuit on $130$-nm CMOS technology.

For both rank-2 SDP and its almost-linear simplification, the dynamical
variables of the Ising model can be regarded as defined on a unit circle.
Except for selected graph families (for example, bipartite
graphs~\cite{burerRankTwo2002}), the machine settles in a state without a
fixed a priori known distribution of the dynamical variables. Consequently,
the machine's terminal state must be \emph{rounded} to recover a feasible
spin configuration. This can be done by comparing dynamical variables with
a selected direction on the unit circle. The spin distribution obtained
this way produces cut that, generally, depends on the choice of the
rounding direction. Finding the optimal rounding, the direction yielding
the highest cut is a separate optimization problem. Its known solutions
(see, e.g. Chapter~8 in~\cite{punnenQuadratic2022}) require non-dynamical
operations. As a result, optimal rounding relies on external processing
power, which makes the relaxation-based Ising machines non-self-contained.

In the present paper, we show that the problem of optimal rounding of
states associated with representing spins by 2D unit vectors is eliminated in
an Ising machine based on a special dynamical system, which we call the
\mdII{} model. More precisely, starting from an arbitrary non-binary state,
the machine settles in a state producing a cut with the weight at least as
large as produced by the optimal rounding of the initial state. We show
that while the considered machine does not necessarily settle in a binary
state, the binary states obtained by rounding yield the same cut regardless
of the choice of the rounding direction. Thus, the terminal states of the
\mdII{}-machine round trivially.

The importance of these findings is two-fold. First, we show that
relaxation-based Ising machines can be self-contained. Second, we
demonstrate that dynamical systems can directly perform complex information
processing tasks.

The rest of the paper is organized as follows. In
Section~\ref{sec:formal-intro}, we remind the main definitions and define
the models of the main interest. In Section~\ref{sec:main-theorems}, we
prove the main results. In Section~\ref{sec:numerics}, we present the
results of numerical simulations.

\section{Continuous dynamical realizations of the Ising model}
\label{sec:formal-intro}

The classical Ising model considers ensembles of coupled binary spins
($\sigma_m \in \left\{ -1, 1 \right\}$ with $m = 1, \ldots, N$). It can be regarded as
a set of binary numbers on the nodes of finite graph
$\sset{G} = \{\sset{V}, \sset{E}\} $, whose edges indicate coupling between
the spins. In the following, we will represent spin distributions as
vectors $\boldsymbol{\sigma} \in \mathbb{R}^N$ and assume that graph $\sset{G}$ is
connected.

Each distribution $\boldsymbol{\sigma}$ is assigned the energy (the cost)
\begin{equation}\label{eq:Ising-Hamiltonian}
  \mathcal{H}(\boldsymbol{\sigma}) = \frac{1}{2}\sum_{m,n} A_{m,n} \sigma_m \sigma_n,
\end{equation}
where $A_{m,n}$ is the graph adjacency matrix. Our main results
(Theorems~\ref{thm:extremal-manifolds}--\ref{thm:optimal-rounding}) hold for
arbitrarily weighted adjacency matrices. However, to simplify the
discussion, we generally assume that the adjacency matrix is
$\left\{ 0,1 \right\}$-weighted, that is
$A_{m,n} \in \left\{ 0, 1 \right\}$.

The computational significance of the Ising model stems from the
observation that its ground state ($\boldsymbol{\sigma}$ delivering the lowest
$\mathcal{H}(\boldsymbol{\sigma})$) solves the maximum cut
problem~\cite{barahonaComputational1982}. Indeed, any distribution
$\boldsymbol{\sigma}$ defines a partitioning of the graph nodes
$\sset{V} = \sset{V}_+ \cup \sset{V}_-$, where $\sset{V}_+$ and
$\sset{V}_-$ are subsets where $\boldsymbol{\sigma}$ is positive and negative,
respectively. The size of the cut is the number (the total weight) of edges
connecting nodes in $\sset{V}_+$ and $\sset{V}_-$. To evaluate the cut size, we
introduce the counting function, which is equal to $1$, if the edge is cut,
and $0$, otherwise. In terms of the spins incident to the edge, such a
function can be written as
\begin{equation}\label{eq:cut-counting-function}
 \Phi(\sigma_m, \sigma_n) = \frac{1}{2} \left( 1 - \sigma_m \sigma_n \right).
\end{equation}
Thus, we obtain for the cut size
\begin{equation}\label{eq:cut-size}
  \mathcal{C}(\boldsymbol{\sigma}) = \frac{1}{2} \sum_{m,n} A_{m,n} \Phi(\sigma_m, \sigma_n) 
  = \frac{M}{2} - \frac{\mathcal{H}(\boldsymbol{\sigma})}{2},
\end{equation}
where $M = \abs{\sset{E}}$ is the number of graph edges, so that the max-cut
problem can be presented as
$\overline{\mathcal{C}}_{\sset{G}} = \max_{\boldsymbol{\sigma} \in \left\{ -1, 1 \right\}^N }
\mathcal{C}(\boldsymbol{\sigma})$.

Finding the maximum cut is an NP-complete
problem~\cite{karpReducibility1972, gareySimplified1976}. Therefore,
other NP-complete problems can be solved by finding the ground states of
Ising models with specially constructed Hamiltonians, as was demonstrated
in~\cite{lucasIsing2014}. We will consider the cut size as our main
objective function. This simplifies the discussion without changing the
essence of the problem.

While the Ising model is inherently discrete, the problem of its ground
state is easy to reformulate in a continuous form that can be realized in a
continuous dynamical system. We will use this observation to illustrate the
operational principles of the dynamical Ising machines and the challenges
associated with the quality of solutions.

To construct a basic continuous Ising machine, we consider function
$\mathcal{C}_{C}(\boldsymbol{\xi})$ with $\boldsymbol{\xi} \in {[-1, 1]}^N$ obtained from
$\mathcal{C}(\boldsymbol{\sigma})$ by substituting $\sigma_{m} \to \xi_{m}$ in
Eqs.~\eqref{eq:cut-counting-function}~and~\eqref{eq:cut-size}. Since
$A_{m,m} = 0$, function $\mathcal{C}_{C}(\boldsymbol{\xi})$ is linear with respect to
all $\xi_{m}$ and, therefore, does not have isolated critical points inside
the cube ${[-1, 1]}^{N}$. In other words, when
$\mathcal{C}_{C}(\boldsymbol{\xi})$ is maximized over ${[-1, 1]}^{N}$, it reaches the
maximal values at the vertices of the cube, where
$\mathcal{C}_{C}(\boldsymbol{\xi})$ coincides with $\mathcal{C}(\boldsymbol{\sigma})$.

The dynamical system governed by
$\dot{\xi}_m = \partial \mathcal{C}_{C}(\boldsymbol{\xi})/\partial \xi_m$ and constrained by
$\xi_{m} \leq 1$ ensures monotonously increasing
$\mathcal{C}(\boldsymbol{\xi}(t))$. Without going into a detailed analysis, we note
that the system's evolution terminates at critical points of
$\mathcal{C}_{C}(\boldsymbol{\xi})$. It is not difficult to show that the states
reached by this system starting from generic initial conditions are not
necessarily binary (with $\xi_{m} = \pm 1$). However, as follows from the
linear property of $\mathcal{C}_{C}(\boldsymbol{\xi})$, those $\xi_{m}$ that did not
terminate at the boundaries of the interval $[-1, 1]$ can be chosen either
$+1$ or $-1$: both choices will produce binary states yielding the same
cut. When recovering a feasible spin configuration from a non-binary state
does not require special processing, we will say that the state rounds
trivially.

Thus, the dynamical system defined by $\mathcal{C}_{C}(\boldsymbol{\xi})$ operates as
an Ising machine. By construction, starting from a generic initial state,
the machine will evolve towards increasing
$\mathcal{C}_{C}(\boldsymbol{\xi})$. Moreover, states solving problem
$\overline{\boldsymbol{\xi}}  
= \arg\max_{\boldsymbol{\xi} \in {[-1,1]}^{N}}
\boldsymbol{C}_{C}(\boldsymbol{\xi})$ trivially round to spin configurations
yielding the maximum cut of graph $\sset{G}$. However, reaching the global
maximum of $\mathcal{C}_{C}(\boldsymbol{\xi})$ from a generic initial state is not
guaranteed. The dynamical system may encounter and terminate at a local
maximum. It is not difficult to show that the only condition imposed on the
machine's terminal state $\boldsymbol{\xi}(\infty)$ reached from a generic
initial state is the following. Let $\boldsymbol{\sigma}$ be the spin
configuration recovered from $\boldsymbol{\xi}(\infty)$. Then, for each node
$m \in \sset{V}$ at least half of the incident edges are
cut~\cite{erementchoukComputational2022}:
\begin{equation}\label{eq:NMR}
  F_m = \sum_{n} A_{m,n} \sigma_m \sigma_n \leq 0.
\end{equation}
Here, $F_m$ has the meaning of the difference between the weights of uncut
and cut edges incident to node $m$. It is worth noting that
$\xi_{m}(\infty) \notin \left\{ -1, 1 \right\} $ if and only if $F_{m} = 0$. Indeed,
if $F_m \ne 0$, one of the points $\xi_m = \pm 1$ is attracting and the other is
repelling. As a result, the system cannot terminate in a state with
$\xi_m(\infty)$ inside the interval $(-1, 1)$. 

Thus, the terminal states of the dynamical model defined by maximizing
$\mathcal{C}_{C}(\boldsymbol{\xi})$ are determined by the same conditions as the
outcome of the $1$-opt local search. This algorithm formulated in terms of
spin variables works as follows. It checks that for all nodes
condition~\eqref{eq:NMR} holds. If for some $m \in \sset{V}$, one has
$F_{m} > 0$, then the respective spin is reverted:
$\sigma_{m} \to -\sigma_{m}$, which reverts the sign of $F_m$. Such inversion increases
the cut weight and since the maximum cut weight is finite, the algorithm
terminates.

Consequently, the introduced machine can be qualitatively regarded as a
dynamical realization of the 1-opt local search. Thus, while the machine
may perform well on some instances of the max-cut problem, its polynomial
time performance is characterized by the approximation ratio $1/2$. In
particular, this implies that on ``challenging'' instances, the machine
based on the direct continuation of $\mathcal{C}(\boldsymbol{\sigma})$ may perform
similarly to random partitioning.

An alternative approach to a dynamical reformulation of the Ising model
stems from systematically adapting the relaxations of the max-cut problem.
To present this approach in a unified manner, we write the objective
function in the form
\begin{equation}\label{eq:obj-fun-max-cuts}
  \mathcal{C}_{M}\left( \boldsymbol{\xi} \right) = \frac{1}{2}
  \sum_{m,n} A_{m,n} \Phi_M \left( \xi_m - \xi_n \right) ,
\end{equation}
where $M$ denotes the model defined by the core function $\Phi_M$. For the
max-cut problem, $M=I$, that is the objective function is defined on binary
variables $\boldsymbol{\xi} = \boldsymbol{\sigma} \in \left\{ -1, 1 \right\}^N$ and
$\Phi_I(\xi) = \xi^2/4$. The graph of $\Phi_I(\xi)$ consists of three points (see
Fig.~\ref{fig:core-functions}). Various relaxations can be constructed as
interpolations of $\Phi_I$ by continuous functions $\Phi : [-2, 2] \to \mathbb{R}$
periodically continued to functions on $\mathbb{R}$. Given the relaxation, the
equations of motion governing the dynamical system are defined requiring
that the objective function increases:
$\dot{\xi}_{m} = \partial \mathcal{C}_{M}(\boldsymbol{\xi})/\partial \xi_{m}$.

\begin{figure}[tb]
  \centering
  \includegraphics[width=2.8in]{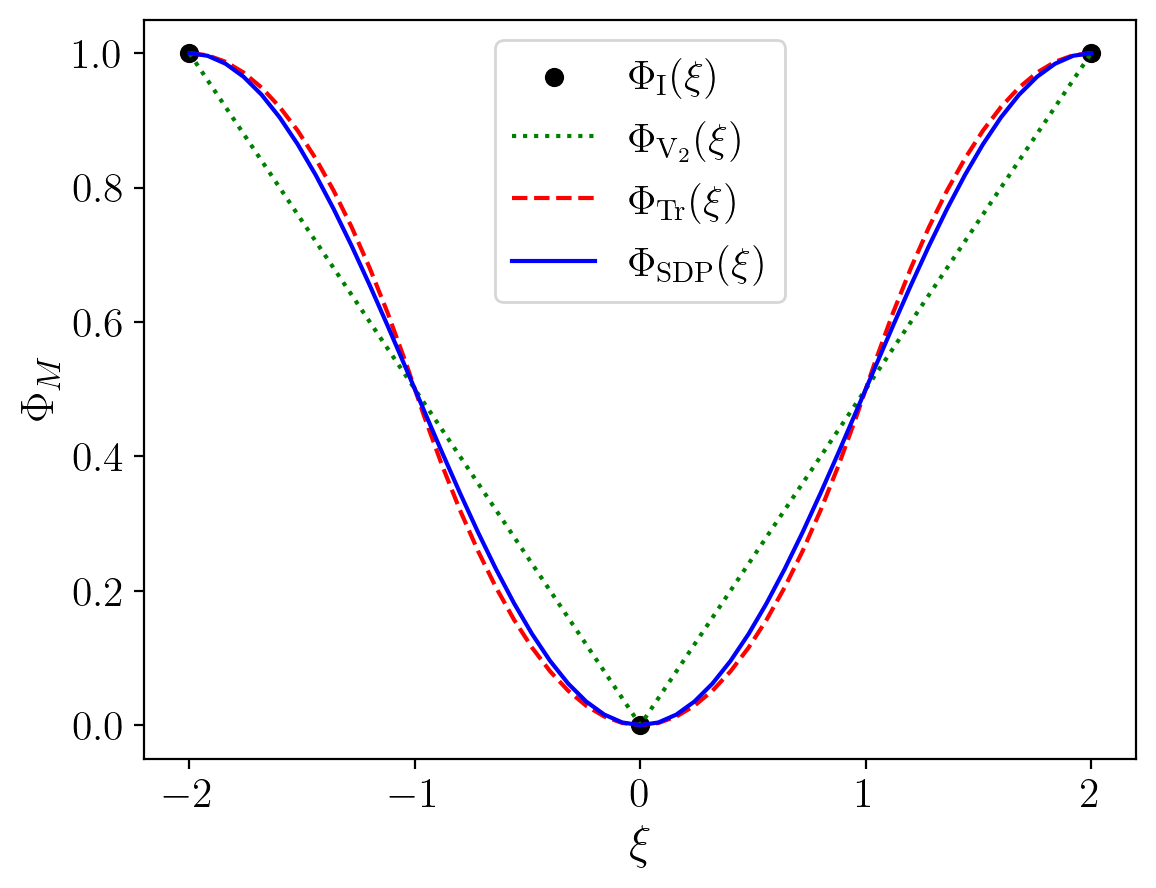}
  \caption{Comparison of core functions of several relaxation-based Ising
    machines: $\Phi_{\mathrm{I}}$ is the discrete function of the binary Ising
  model, $\Phi_{\mathrm{SDP}}(\boldsymbol{\xi})$,
  $\Phi_{\mathrm{Tr}}(\boldsymbol{\xi})$, and $\Phi_{\mathrm{GW}}(\boldsymbol{\xi})$
  are the core functions of the rank-$2$ SDP relaxation, the triangular
  model from Refs.~\cite{shuklaScalable2022, shuklaCustom2023}, and the
  \mdII{} model investigated in the present paper.}
  \label{fig:core-functions}
\end{figure}

Among different relaxations, we emphasize three. The first one corresponds
to the rank-$2$ SDP relaxation~\cite{burerRankTwo2002}, which is given by
$\Phi_{\text{SDP}}(\xi) = \left( 1 - \cos \left( \pi \xi/2 \right) \right)/2$. Its
piece-wise parabolic continuous approximation yields the model investigated
in~\cite{shuklaScalable2022}, where it was dubbed the triangular model.
Near $\xi = 0$ (for $\abs{\xi} \leq 1$), it is defined by
$\Phi_{\text{Tr}}(\xi) = \xi^2/2$, and near $\abs{\xi} = 2$ by
$\Phi_{\text{Tr}}(\xi) = 1 - (\abs{\xi}-2)^2/2$. Finally, the model of the main
interest of the present paper is given for $\abs{\xi} \leq 2$ by
$\Phi_\IIind(\xi) = \abs{\xi}/2$, which is called the \mdII{} model owing to the
shape of $\Phi_\IIind(\xi)$. Respectively, we will call the \mdII-machine the
Ising machine with the dynamics determined by this model. The core
functions of these models are compared in Fig.~\ref{fig:core-functions}.

\section{V${}_2$ model}\label{sec:main-theorems}

\subsection{Model definition}

Assuming that $-2 \leq \xi_m \leq 2$, we can write the objective function of the
model with the core function $\Phi_\IIind(\xi)$ as
\begin{equation}\label{eq:mdII-objective}
  \mathcal{C}_{\IIind}(\boldsymbol{\xi}) = \frac{1}{4} \sum_{m,n} A_{m,n} \abs{\xi_m - \xi_n}_{[-2,2]},
\end{equation}
where $\abs{\ldots}_{[-2,2]}$ is a periodic function with period $P=4$ defined by
$\abs{\xi}_{[-2,2]} = \abs{\xi}$ for $-2 \leq \xi \leq 2$. Alternatively,
$\abs{\ldots}_{[-2,2]}$ can be defined as the distance on a circle with
circumference $4$. Such a definition connects $\mathcal{C}_{\IIind}(\boldsymbol{\xi})$
with a model based on spins represented by unit two-dimensional vectors
$\vec{s}_m$:
\begin{equation}\label{eq:mdII-objective-vector}
  \mathcal{C}_{\IIind}(\boldsymbol{s}) = \frac{1}{2\pi} \sum_{m,n} A_{m,n} \arccos \left( \vec{s}_m
  \cdot \vec{s}_n\right).
\end{equation}
Indeed, defining vectors $\vec{s}_m$ by their polar angles $\pi \xi_m/2$
turns~\eqref{eq:mdII-objective-vector} into~\eqref{eq:mdII-objective}.
A model similar to~\eqref{eq:mdII-objective-vector} (with $N$-dimensional
$\vec{s}_m$) was introduced by Goemans and Williamson
in~\cite{goemansImproved1995}, where it appeared naturally as the average
value of the cut produced by rounding the SDP solution relatively to random
hyperplanes. 

This model has a property that distinguishes it in the family of dynamical
Ising machines. Starting from a non-binary state, the \mdII-machine evolves
to a state that trivially rounds to a spin configuration producing a cut at
least as large as that obtained by the best rounding of the starting state
(see Theorem~\ref{thm:optimal-rounding}). An important consequence of this
property is that the \mdII-machine can be used for rounding and simple
post-processing of the output of relaxation-based Ising machines, as
illustrated by Fig.~\ref{fig:round-flow}. Therefore, a dynamical system
driving the Ising machine can be chosen freely, for instance, on the ground
of efficiency to solve particular instances of optimization problems.

\begin{figure}[tb]
  \centering
  \includegraphics[width=3in]{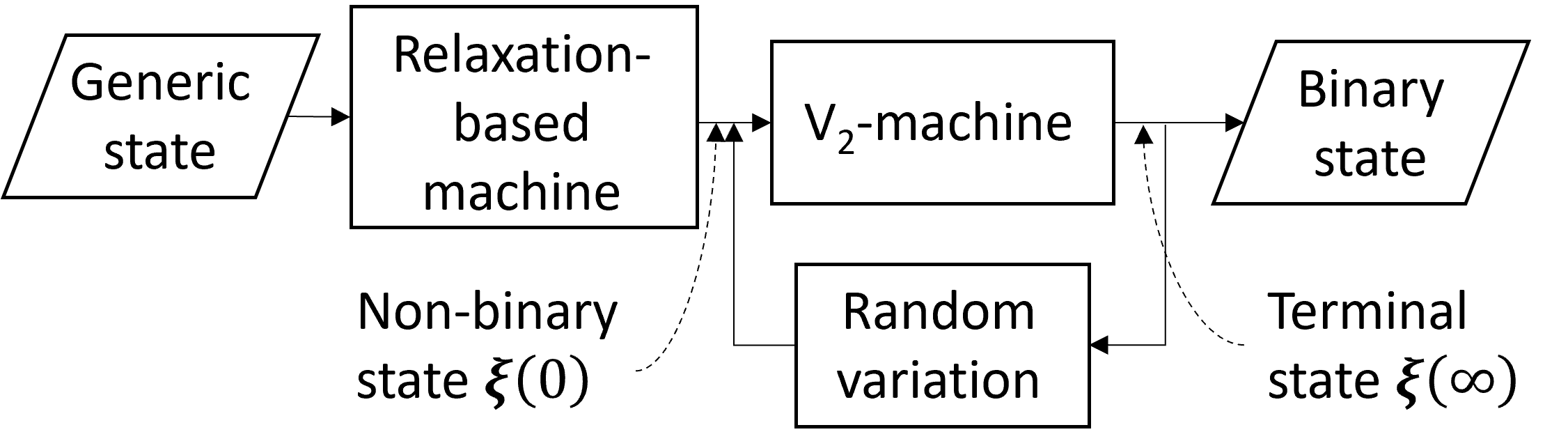}
  \caption{A relaxation-based heterogeneous Ising machine returning a
    binary state. First, a dynamical system based on a relaxation (for
    instance, rank-2 SDP) is initiated by a generic state. Then, the
    dynamical core is switched to the \mdII{} model, which delivers the
    optimal rounding of the state obtained during the first stage and
    performs basic post-processing by small perturbations of the terminal
    state (see Theorem~\ref{thm:mdii-non-worsening}).}
  \label{fig:round-flow}
\end{figure}

\subsection{Model dynamics}

The evolution of the dynamical model realizing the \mdII-machine is governed by
$\dot{\boldsymbol{\xi}} = \nabla_{\boldsymbol{\xi}} \mathcal{C}_{\IIind}(\boldsymbol{\xi})$, which
ensures that $\mathcal{C}_{\IIind}(\boldsymbol{\xi})$ monotonously increases with time.
Thus, for individual nodes, one has
\begin{equation}\label{eq:mdII-eq-motion}
  \dot{\xi}_m = \sum_{n} A_{m,n} \phi_{\IIind}(\xi_m - \xi_n),
\end{equation}
where $\phi_{\IIind}(\xi)$ is a periodic function, which inside the period,
$\xi \in (-P/2, P/2]$,
is defined by $\phi_{\IIind}(\xi) = \sgn(\xi)/2$ for
$\xi \in (-P/2, P/2]$, and $\sgn(\xi)$ is the sign function
\begin{equation}\label{eq:sign-function}
\operatorname{sgn}(\xi) =
\begin{cases}
-1, & \text{if } \xi < 0, \\
0, & \text{if } \xi = 0, \\
1, & \text{if } \xi > 0.
\end{cases}
\end{equation}
Similar equations of motion were empirically studied in the context of the
phase dynamics in a saturated Kuramoto oscillatory neural network
in~\cite{delacourMixedsignal2023}.

It must be noted that
the dynamics of the \mdII-machine requires special attention at the
switching surfaces, $\xi_{m} = \xi_{n}$ for $A_{m,n} \ne 0$, at which
$\phi_{\IIind}\left( \xi_{m} - \xi_{n} \right)$ is discontinuous. The behavior of
discontinuous dynamical systems with switching determined by the dynamical
variables is a complex problem (see for reviews
Refs.~\cite{filippovDifferential1988, cortesDiscontinuous2008}). However,
for the \mdII{} model the problem significantly simplifies because the
dynamics is subject to maximizing $\mathcal{C}_{\IIind}(\boldsymbol{\xi})$.
To avoid obscuring the most interesting features
of the \mdII{} model, which justify a detailed investigation of the model,
we limit ourselves to enforcing the convention $\sgn(0) = 0$. We will
provide a rigorous consideration of the \mdII{} model within the framework
of discontinuous dynamical systems elsewhere.

\subsection{Basic properties of terminal states}
\label{sec:basic-terminals}

Starting from the initial state $\boldsymbol{\xi}(0)$, the machine traverses
the trajectory $\boldsymbol{\xi}(t \vert \boldsymbol{\xi}(0))$ until it reaches the
terminal state $\boldsymbol{\xi}(\infty \vert \boldsymbol{\xi}(0))$ characterized by
$\dot{\boldsymbol{\xi}} = 0$. Thus, the machine's terminal state is a
critical point of $\mathcal{C}_{\IIind}(\boldsymbol{\xi})$. In contrast to the
machines, whose dynamics is determined by $\mathcal{C}_{\text{SDP}}$ and
$\mathcal{C}_{\text{Tr}}$, the \mdII-machine reaches equilibrium in finite time.
Indeed, for out-of-equilibrium states, the rate of changing of
$\mathcal{C}_{\IIind}(\boldsymbol{\xi})$ is limited from below. For example, for
$\left\{ 0,1 \right\}$-weighted graphs, one has
\begin{equation}\label{eq:mdII-out-of-eq}
  \frac{d}{dt} \mathcal{C}_{\IIind}(\boldsymbol{\xi}) =
  \sum_{m} \left( \frac{\partial \mathcal{C}_{\IIind}}{\partial \xi_m} \right)^2 \geq 1.
\end{equation}
As a result, the time needed by the \mdII-machine to reach equilibrium is
limited from above by
$\overline{\mathcal{C}}_{\IIind} = \max_{\boldsymbol{\xi}} \mathcal{C}(\boldsymbol{\xi})$.

Using the same argument as in~\cite{goemansImproved1995}, it can be proven
that $\overline{\mathcal{C}}_{\IIind} = \overline{\mathcal{C}}_{\sset{G}}$, or, in other words,
that $\mathcal{C}_{\IIind}(\boldsymbol{\xi})$ is an \emph{exact} relaxation (see also
Ref.~\cite{steinerbergerMaxCut2023}). Indeed, on the one hand, one has
$\overline{\mathcal{C}}_{\IIind} \geq \overline{\mathcal{C}}_{\sset{G}}$, since the \mdII\ model
is a relaxation. On the other hand, we observe that the probability for a
random point on interval $[-2,2]$ to get between given $\xi_m$ and $\xi_n$ is
$\abs{\xi_m - \xi_n}/4$. Thus, $\mathcal{C}_{\IIind}(\boldsymbol{\xi})$ can be regarded as
the average size of cut obtained by random rounding of $\boldsymbol{\xi}$.
Since the size of cut produced by an individual rounding cannot exceed
$\overline{\mathcal{C}}_{\sset{G}}$, we obtain that
$\overline{\mathcal{C}}_{\sset{G}} \geq \overline{\mathcal{C}}_{\IIind}$.

However, the fact that the global maximum of
$\mathcal{C}_{\IIind}(\boldsymbol{\xi})$ is the maximum cut of the graph is not sufficient
for our purposes, as we are primarily interested in finding partitions
delivering (approximately) the maximum cut rather than the weight of the
maximum cut. Therefore, we need to consider
the structure of states of the \mdII-machine.

We notice that $\mathcal{C}_{\IIind}(\boldsymbol{\xi})$ is invariant with respect to
global translations
$\boldsymbol{\xi} \to \boldsymbol{\xi} + a \boldsymbol{1}$, where $a$ is a real
number and $(\boldsymbol{1})_m = 1$. Therefore, we need to identify as
binary all states that are obtained by displacing some
$\boldsymbol{\sigma}$. Even in view of the extended definition of binary states,
the fact that a representation is exact does not imply that all states
delivering maxima of $\mathcal{C}_{\IIind}$ are binary. It is, therefore, important
that the \mdII{} model in addition to being exact has a stronger property:
critical points of $\mathcal{C}_{\IIind}(\boldsymbol{\xi})$ are at least in the same
connected manifolds as binary states. Thus, the critical values of
$\mathcal{C}_{\IIind}$ coincide with possible values of cuts.
\begin{theorem} \label{thm:extremal-manifolds}
  Let
  $\sset{M}\left( c \right) = \left\{ \boldsymbol{\xi} \in \mathbb{R}^N :
    \nabla_{\boldsymbol{\xi}}\mathcal{C}_{\IIind}(\boldsymbol{\xi}) = 0,
    \mathcal{C}_{\IIind}(\boldsymbol{\xi}) = c \right\}$ be the manifold of critical
  points corresponding to the same critical value $c$, then each connected
  component of $\sset{M}(c)$ contains a binary state.
\end{theorem}

Before we turn to the proof of this theorem, we introduce new dynamical
variables. Any number $\xi \in \mathbb{R}$ can be uniquely presented as
\begin{equation}\label{eq:mdII-xi-rep}
  \xi = \sigma + X + 4k,
\end{equation}
where $\sigma \in \left\{ -1, 1 \right\}$, $X \in (-1, 1]$, and
$k \in \mathbb{Z}$. The last term, representing multiples of the period of the
counting function, will play the minor role, and, therefore, we will also
write $\xi = \sigma + X \mod P$, which should be understood in the sense of
Eq.~\eqref{eq:mdII-xi-rep}. Because of the importance of this
representation, we will call the pair $(\sigma, X)$ \emph{a relaxed spin} and
refer to $\sigma$ and $X$ as its discrete and continuous components,
respectively.

Using~\eqref{eq:mdII-xi-rep} in
Eq.~\eqref{eq:mdII-objective}, we can write the relaxed spin representation
for the objective function
\begin{equation}\label{eq:mdII-cut-sx-rep}
  \mathcal{C}_{\IIind} \left( \boldsymbol{\xi} \right) =
  \mathcal{C}_{\IIind} \left( \boldsymbol{\sigma}, \mathbf{X} \right) :=
  \mathcal{C}\left( \boldsymbol{\sigma} \right)
  + \widetilde{\mathcal{C}}_{\IIind}\left( \boldsymbol{\sigma}, \mathbf{X} \right) ,
\end{equation}
where
\begin{equation}\label{eq:mdII-cut-add}
  \widetilde{\mathcal{C}}_{\IIind}\left( \boldsymbol{\sigma}, \mathbf{X} \right)
  = \frac{1}{4}\sum_{m,n} A_{m,n} \sigma_m \sigma_n \abs{X_m - X_n}.
\end{equation}
Equation~\eqref{eq:mdII-cut-sx-rep} can be derived by noticing that
\begin{equation}\label{eq:mdII-Phi-diff}
  \Phi_{\IIind} (X_m - X_n + \sigma_{m} - \sigma_{n}) = \frac{1}{2}
  \left(1 - \sigma_m \sigma_n\right) + \sigma_m \sigma_n \Phi_{\IIind}\left( X_m - X_n \right) .
\end{equation}
This equality obviously holds when $\sigma_m = \sigma_n$. In turn, when
$\sigma_m - \sigma_n = \pm2$, the equality follows from the symmetry of the counting
function $\Phi_{\IIind}(x \pm 2) = 1 - \Phi_{\IIind}(x)$. It is worth noting that the
counting functions $\Phi_{\text{SDP}}$ and $\Phi_{\text{Tr}}$ also have such a symmetry and,
therefore, representations similar to~\eqref{eq:mdII-cut-sx-rep} can be
obtained for rank-$2$ SDP and triangular model as well. Finally,
Eq.~\eqref{eq:mdII-cut-add} is written considering that
$\Phi_{\IIind} \left( X_m - X_n \right) = \abs{X_m - X_n}/2$.

The equations of motion describing the dynamics of the \mdII-machine in
terms of the new variables have the form
\begin{equation}\label{eq:mdII-eqmo-sx}
  \dot{X}_m = \frac{1}{2}\sum_{n} A_{m,n} \sigma_m \sigma_n \sgn\left( X_m - X_n \right) .
\end{equation}
They should be solved while taking into account the topology of the new
variables (see Fig.~\ref{fig:sx-topology}). For example, the numerical
simulations presented in the next section updated the dynamical variables
as described in procedure \textproc{Update} listed in
Section~\ref{sec:numerics}.

\begin{figure}[tb]
  \centering
  \includegraphics[width=1.5in]{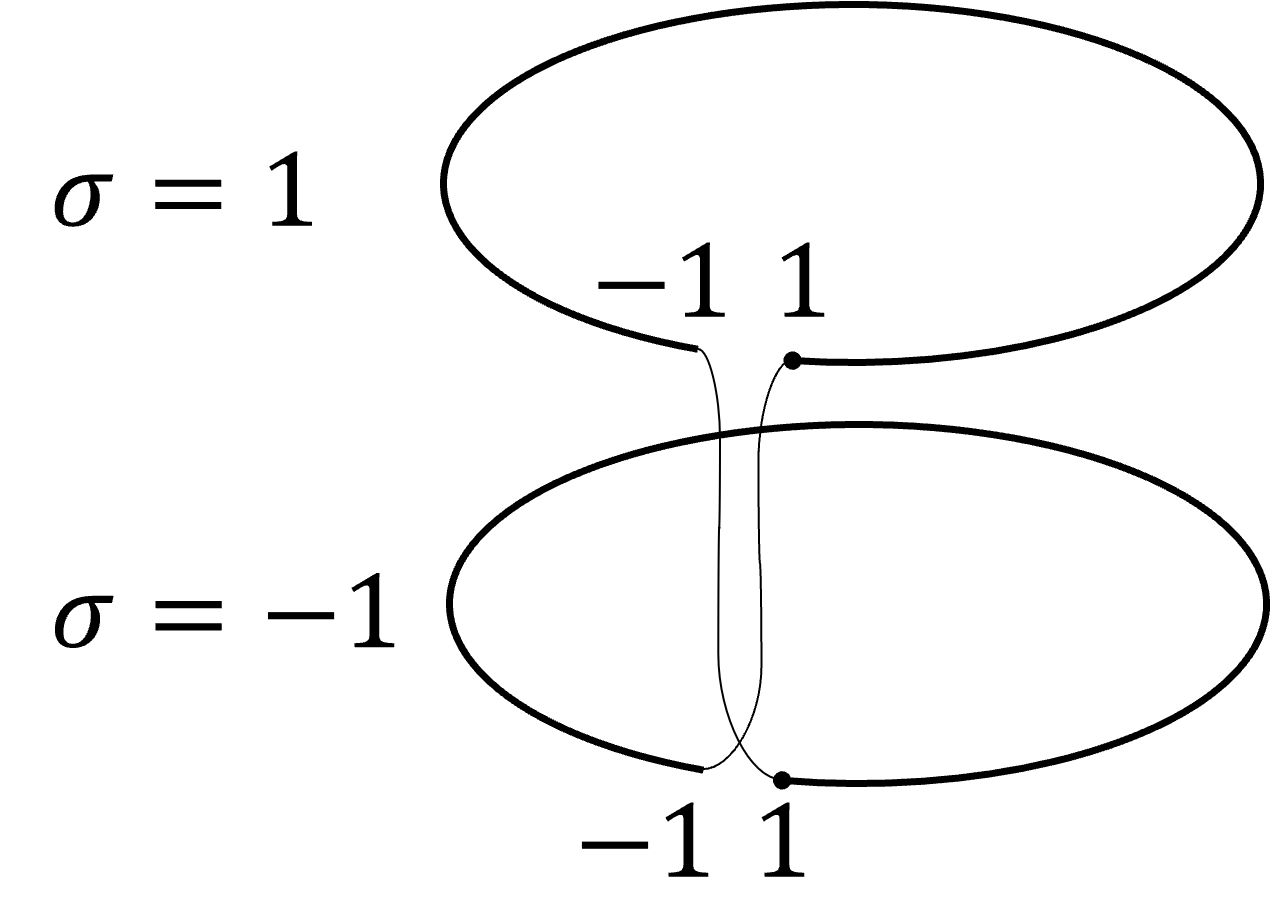}
  \caption{The topology of the representation $\xi = \sigma + X \mod P $. The bold
  circles represent intervals $(-1,1]$, and the thin lines indicate the
  transitions at the boundaries of these intervals.}
  \label{fig:sx-topology}
\end{figure}

In the new variables, it is apparent  that the lack of isolated non-binary
critical points is the consequence of the local linearity of
$\mathcal{C}_{\IIind}(\boldsymbol{\sigma}, \mathbf{X})$.

\begin{proof}[Proof of Theorem~\ref{thm:extremal-manifolds}]
  Let $\boldsymbol{\xi}$ be a critical point of $\mathcal{C}_{\IIind}(\boldsymbol{\xi})$
  that is not a displacement of a binary state. In other words, one has
  $\xi_m = \sigma_m + X_m \mod P$ and not all $X_m$ are the same (without loss of
  generality, we
  can assume that $X_m \ne 1$ for all $m$). We
  show that $\boldsymbol{\xi}$ can be contracted to a binary state while
  staying on the critical manifold. 

  We partition $\sset{V}$ by collecting nodes with the same $X_m$, so that
  $\sset{V} = \bigcup_{p=1}^S\sset{V}^{(p)}$, with $S \leq N$, where
  $\sset{V}^{{(p)}} = \left\lbrace m \in \sset{V} : \xi_{m} = X^{(p)} + \sigma_{m} \mod P
  \right\rbrace$ and $X^{(p)} \ne X^{(q)}$ for $p \ne q$.
  Enumerating $\sset{V}^{(p)}$ in such a way that $X^{(p)} > X^{(q)}$ for
  $p > q$, we can rewrite Eq.~\eqref{eq:mdII-cut-sx-rep} as
  \begin{equation}\label{eq:mdII-partition-xp}
    \mathcal{C}_{\IIind}(\boldsymbol{\xi}) = \mathcal{C}(\boldsymbol{\sigma})
    + \frac{1}{2} \sum_{p > q}
    \sum_{\substack{m \in \sset{V}^{(p)} \\ n \in \sset{V}^{(q)} }}
    A_{m,n} \sigma_m \sigma_n \left( X^{(p)} - X^{(q)}\right).
  \end{equation}
  With respect to each $X^{(p)}$, this is a linear function, hence, its
  gradient does not depend on the magnitude of $\mathbf{X}$. Let $\left\{ \lambda_1, \lambda_2 \ldots
  \right\}$ with $0 < \lambda_j < 1$ be a Cauchy sequence converging to 0. Since
  $\lambda_j X^{(p)} > \lambda_j X^{(q)}$ for $p > q$, the sequence
  $\boldsymbol{\xi}_j= \boldsymbol{\sigma} + \lambda_j\mathbf{X}$ is a Cauchy sequence
  of critical points of $\mathcal{C}_{\IIind}$ converging to $\boldsymbol{\sigma}$.
\end{proof}

The linear form of Eq.~\eqref{eq:mdII-partition-xp} results in a
clustered form of the terminal states of the \mdII-machine. Let
$\boldsymbol{\xi} = \boldsymbol{\sigma} + \mathbf{X} \mod P$ be
a critical point of $\mathcal{C}_{\IIind}(\boldsymbol{\xi})$. Similar to above, we
introduce partitioning of the graph nodes
$\sset{V} = \bigcup_{p = 1}^S \sset{V}^{(p)}$, where $\sset{V}^{(p)}$ is a subset
of nodes with coinciding $X$-coordinates, $X^{(p)}$, so that if node $m$
belongs to some $\sset{V}^{(p)}$, then
$\xi_m = \sigma_m + X^{(p)} \mod P$. When the number of parts is
strictly less than the number of nodes, $S < N$, some $\sset{V}^{(p)}$
contain more than one node. As we will see shortly, terminal states with
$S < N$ are rather common. To reflect this observation, we will call
$\sset{V}^{(p)}$ \emph{clusters}, even if $\abs{\sset{V}^{(p)}} = 1$ for
all $p$, in which case the clusters can be regarded as \emph{trivial}. The
notion of clusters is illustrated by Fig.~\ref{fig:clusters} that depicts a
``generic'' critical point $\boldsymbol{\xi}$ on the circle with
circumference $P$.

\begin{figure}[tb]
  \centering
  \includegraphics[width=0.7\textwidth]{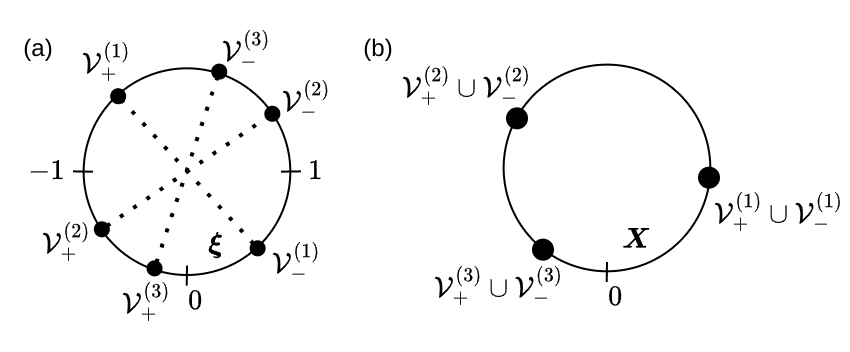}
  \caption{The clustered structure of critical points of
    $\mathcal{C}_{\IIind}$. (a) A typical form of
    $\boldsymbol{\xi}$ with bold points indicating $\xi_n$ at nodes in sets
    forming a partition
    $\sset{V} = \sset{V}^{(1)} \cup \sset{V}^{(2)} \cup \sset{V}^{(3)}$. Each
    cluster is made of two components,
    $\sset{V}^{(p)} = \sset{V}^{(p)}_- \cup \sset{V}^{(p)}_+$, 
    corresponding to the respective values of $\sigma_m$. (b) The same state
    shown for the $(\boldsymbol{\sigma}, \mathbf{X})$ representation (relaxed spins).}
  \label{fig:clusters}
\end{figure}

Details of the distribution of the nodes over the clusters and the spin
states within the clusters non-trivially depend on the structure of the
graph weighted adjacency matrix. A comprehensive consideration of these
properties goes beyond the scope of the present paper, and, therefore, we
limit ourselves to several simple examples. Let $\boldsymbol{\xi}$ be a
critical point with clusters $\sset{V}^{(p)}$ enumerated in such a way that
$X^{(p)} > X^{(q)}$ for $p > q$. Then the equilibrium condition for node $m
\in \sset{V}^{(p)}$ can be written as
\begin{equation}\label{eq:cluster_p_equilibrium}
  \sum_{q < p} \sum_{n \in \sset{V}^{(p)}} A_{m,n} \sigma_n - 
  \sum_{q > p} \sum_{n \in \sset{V}^{(p)}} A_{m,n} \sigma_n = 0.
\end{equation}
Writing this condition for all graph nodes, we obtain a system of
homogeneous equations with respect to spin variables $\sigma_n$. This system has
feasible solutions ($\sigma_m \in \left\{ -1, 1 \right\}$) only if each row
of the graph weighted adjacency matrix has elements that are
linearly-independent over $\left\{ - 1, 1 \right\}$. This implies that for
generic adjacency matrices with real matrix elements, terminal states can
consist of only one cluster, when the left-hand-side of
Eq.~\eqref{eq:cluster_p_equilibrium} vanishes identically. In other words,
all terminal states have the form of a homogeneously displaced binary state
with $\xi_m = \sigma_m + X \mod P$. This justifies the remark above
about the clustered form of terminal states being common.
Figure~\ref{fig:clusters-dyn} demonstrates the emergence of clusters for
the example of complete graph $\sset{K}_{12}$. More generally, as will
become apparent from the proof of Theorem~\ref{thm:r-invariant-rounding}
below the number of clusters in the terminal state cannot exceed the number
of binary states in the containing it critical manifold.


\begin{figure}[tb]
  \centering
  \includegraphics[width=0.7\textwidth]{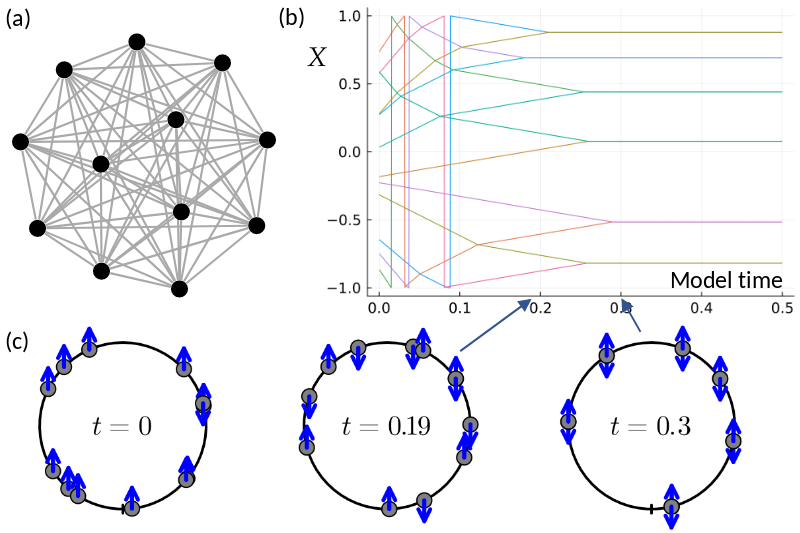}
  \caption{The emergence of clusters in terminal states of the \mdII{}
    model. (a) Graph $\sset{K}_{12}$ used as an example. (b) The time
    dependence of the $X$-component of the relaxed spins. The arrows
    indicate the instances of snapshots in (c). The time is measured in
    units of Eq.~\eqref{eq:mdII-eqmo-sx}. (c) Snapshots of the time
    evolution of the relaxed spins on $\sset{K}_{12}$. }
  \label{fig:clusters-dyn}
\end{figure}

Most of the critical points of $\mathcal{C}_{\IIind}$ are saddle points. In their
vicinity, the shape of $\mathcal{C}_{\IIind}$ is determined by the homogeneity of
$\widetilde{\mathcal{C}}_{\IIind}(\boldsymbol{\sigma}, \mathbf{X})$:
\begin{equation}\label{eq:mdII-scaling-property}
  \mathcal{C}_{\IIind}(\boldsymbol{\sigma}, \lambda \mathbf{X}) = \mathcal{C}(\boldsymbol{\sigma}) +
  \abs{\lambda} \widetilde{\mathcal{C}}_{\IIind}(\boldsymbol{\sigma}, \mathbf{X}).
\end{equation}
It should be noted that while this expression is defined for
$\norm{\mathbf{X}}_\infty \leq 1$ and
$\norm{\lambda\mathbf{X}}_\infty \leq 1$, it formally holds for
$\abs{\lambda} \leq \min_{(m,n) \in \sset{E}}2/\abs{X_m - X_n}$. For example, let
$\boldsymbol{\delta}(\boldsymbol{\sigma}' \vert \boldsymbol{\sigma})$ be a
$(0,1)$-vector pointing from $\boldsymbol{\sigma}$ to $\boldsymbol{\sigma}'$ in the
$\boldsymbol{\xi}$-space, that is
$\delta(\boldsymbol{\sigma}' \vert \boldsymbol{\sigma})_m = 0$, if
$\sigma_m' = \sigma_m$, and
$\delta(\boldsymbol{\sigma}' \vert \boldsymbol{\sigma})_m = 1$, if
$\sigma_m' = -\sigma_m$. Then,
$\boldsymbol{\sigma} + 2 \boldsymbol{\delta}(\boldsymbol{\sigma}' \vert \boldsymbol{\sigma}) =
\boldsymbol{\sigma}' \mod P$. Using this relation
in~\eqref{eq:mdII-scaling-property}, we obtain the correct
\begin{equation}\label{eq:mdII-cut-variation}
  \widetilde{\mathcal{C}}_{\IIind}\left(\boldsymbol{\sigma},
    \boldsymbol{\delta}(\boldsymbol{\sigma}'\vert \boldsymbol{\sigma})  \right)
  = \frac{1}{2} \left[ \mathcal{C}(\boldsymbol{\sigma}') - \mathcal{C}(\boldsymbol{\sigma}) \right].
\end{equation}
This relation can be proven directly by noticing that, in terms of
$\mathbf{X}(\boldsymbol{\sigma}' \vert \boldsymbol{\sigma})$, the relation between
$\boldsymbol{\sigma}'$ and $\boldsymbol{\sigma}$ can be written as
\begin{equation}\label{eq:mdII-sp-s-relation}
  \sigma_m' = \sigma_m \left( 1 - 2 \delta(\boldsymbol{\sigma}' \vert \boldsymbol{\sigma})_m \right) .
\end{equation}
Then, we have the following chain of equalities, where we have omitted the
argument of $\boldsymbol{\delta}(\boldsymbol{\sigma}' \vert \boldsymbol{\sigma})$,
\begin{equation}\label{eq:dyn-mdII-varcuts}
  \begin{split}
    \mathcal{C}(\boldsymbol{\sigma}') - \mathcal{C}(\boldsymbol{\sigma}) = &
    \frac{1}{4} \sum_{m,n} A_{m,n} \left( \sigma_m \sigma_n -
      \sigma'_m \sigma'_n \right) \\
    & = \frac{1}{2} \sum_{m,n} A_{m,n} \sigma_m \sigma_n
    \left(\delta_m + \delta_n - 2 \delta_m\delta_n \right) \\
    & = \frac{1}{2} \sum_{m,n} A_{m,n} \sigma_m \sigma_n
      \left( \delta_m - \delta_n \right)^2 \\
    & = \frac{1}{2} \sum_{m,n} A_{m,n} \sigma_m \sigma_n \abs{\delta_m - \delta_n} \\
    & = 2 \widetilde{\mathcal{C}}_{\IIind}
      \left( \boldsymbol{\sigma}, \boldsymbol{\delta}(\boldsymbol{\sigma}' \vert \boldsymbol{\sigma}) \right),
  \end{split}
\end{equation}
where we have taken into account that for $(0,1)$-vectors ${\delta_m}^2 = \delta_m$.

It follows from Eq.~\eqref{eq:mdII-eqmo-sx} that for any $\mathbf{Y} \in \mathbb{R}^N$ such
that $\norm{\mathbf{Y}}_\infty \leq 1$ and $\lambda \in \mathbb{R}$ such that $\norm{\lambda \mathbf{Y}}_\infty
\leq 1$, one has
\begin{equation}\label{eq:mdII-hom-grad}
  \left. \nabla_{\mathbf{X}} \widetilde{\mathcal{C}}_{\IIind} \left( \boldsymbol{\sigma},
      \mathbf{X} \right)\right\vert_{\mathbf{X} = \lambda \mathbf{Y}}
  =
  \sgn(\lambda)
  \left.\nabla_{\mathbf{X}} \widetilde{\mathcal{C}}_{\IIind} \left( \boldsymbol{\sigma},
      \mathbf{X} \right)\right\vert_{\mathbf{X} = \mathbf{Y}} .
\end{equation}
However, since $\nabla_{\mathbf{X}} \mathcal{C}_{\IIind}$ is discontinuous, this does not
imply that points $\boldsymbol{\xi}(t) = \boldsymbol{\sigma} + t
\boldsymbol{\delta}\left( \sigma' \vert \boldsymbol{\sigma} \right) $ with $0 < t < 2$, that
is lying on the segment connecting $\boldsymbol{\sigma}$ and $\boldsymbol{\sigma}'$,
are critical. For that, we need a more detailed analysis of the internal
structure of the critical points.

The immediate consequence of Eqs.~\eqref{eq:mdII-scaling-property}
and~\eqref{eq:mdII-cut-variation} is
\begin{theorem}\label{thm:max_min_cut}
  All binary states, except for max-cut and $\pm \boldsymbol{1}$, are saddle
  points of $\mathcal{C}_{\IIind}(\boldsymbol{\xi})$.
\end{theorem}
\begin{proof}
  Let $\boldsymbol{\sigma}$ be a state with
  $0 < \mathcal{C}(\boldsymbol{\sigma}) < \overline{\mathcal{C}}_{\sset{G}}$, then there exist states
  $\boldsymbol{\sigma}_\pm$ such that
  $\mathcal{C}(\boldsymbol{\sigma}_-) < \mathcal{C}(\boldsymbol{\sigma})$ and
  $\mathcal{C}(\boldsymbol{\sigma}_+) > \mathcal{C}(\boldsymbol{\sigma})$. Functions
  $C_\pm(t) = \mathcal{C}_{\IIind}(\boldsymbol{\sigma} + t \boldsymbol{\delta}(\boldsymbol{\sigma}_\pm \vert
  \boldsymbol{\sigma}))$ are defined on the interval $-2 < t < 2$ and
  $t = 0$ is the local maximum of $C_-(t)$ and the local minimum of
  $C_+(t)$. Hence, $\boldsymbol{\sigma}$ is a saddle point of
  $\mathcal{C}_{\IIind}(\boldsymbol{\xi})$.

  A similar argument shows that states $\pm \boldsymbol{1}$ are
  the only minima and the max-cut states are the only maxima of
  $\mathcal{C}_{\IIind}(\boldsymbol{\xi})$.
\end{proof}

\subsection{Invariance of the cut size with respect to rounding}

The critical points of $\mathcal{C}(\boldsymbol{\xi})$ do not have to be binary, or, in other
words, the critical manifold do not have to be the union of isolated
points. The presence of non-binary states raises the question about their rounding.
As we will show, this question resolves trivially for the \mdII\ model
written in terms of variables $\boldsymbol{\sigma}$ and $\mathbf{X}$: simple
discarding the continuous component $\mathbf{X}$ yields the best rounded
state.

We start by noticing that Eq.~\eqref{eq:mdII-xi-rep} defines mapping
$\widehat{S} : \boldsymbol{\xi} \mapsto (\boldsymbol{\sigma}, \mathbf{X})$, which
implements rounding by fixing the reference point for $\boldsymbol{\xi}$. Due
to the translational symmetry of $\mathcal{C}_{\IIind}(\boldsymbol{\xi})$, the
reference point can be freely changed leading to a family of mappings
$\widehat{S}_r : \boldsymbol{\xi} \mapsto (\boldsymbol{\sigma}(r), \mathbf{X}(r)) $
defined by
\begin{equation}\label{eq:mdII-r-mappings}
  \xi_m - r = \sigma_m(r) + X_m(r) \mod P.
\end{equation}
Using this relation, we define rounding of an arbitrary state
$\boldsymbol{\xi}$ with respect to the rounding center $r$ as
$\widehat{R}_r : \boldsymbol{\xi} \mapsto \boldsymbol{\sigma}(r)$. It suffices to
consider $0 \leq r < 2$, since
$\boldsymbol{\sigma} + 2 = - \boldsymbol{\sigma} \mod P$ and, therefore,
$\boldsymbol{\sigma}(r + 2) = -\boldsymbol{\sigma}(r)$.

For $0 < r < 1 + \min_m X_m$, the variation of $r$ leaves both terms in
Eq.~\eqref{eq:mdII-cut-sx-rep} intact. For larger values of $r$, however,
some spins in $\boldsymbol{\sigma}(r)$ are reversed comparing to
$\boldsymbol{\sigma}(0)$ and, generally, for an arbitrary non-binary state
$\mathcal{C}(\boldsymbol{\sigma}(r)) \ne \mathcal{C}(\boldsymbol{\sigma}(0))$. It is, therefore, an
important property of the \mdII{} model that the size of the cut, produced by
rounding a non-binary \emph{critical} point, does not depend on $r$.

\begin{theorem}
  \label{thm:r-invariant-rounding}
  Let $c$ be a critical value of $\mathcal{C}_{\IIind}$, $\sset{M}(c)$ be the manifold of
  critical points yielding $c$, and
  $\sset{S}(c) = \left\{ \boldsymbol{\sigma} \in \left\{ -1, 1 \right\}^N :
    \mathcal{C}(\boldsymbol{\sigma}) = c\right\}$ be the set of binary states in
  $\sset{M}(c)$. Then, for all $\boldsymbol{\xi} \in \sset{M}(c)$, one has
  $\widehat{R}_r [\boldsymbol{\xi}] \in \sset{S}(c)$. In other words, 
  rounding a non-binary critical point with respect to different rounding centers
  produces binary states yielding the same cut.
\end{theorem}

\begin{proof}
Let $\left\{ \sset{V}^{(p)} \right\}_{p = 1, \ldots, S} $ be a partition of
graph nodes $\sset{V}$ such that
$\sset{V}^{(p)} = \left\{ m \in \sset{V} : \xi_m = X^{(p)} + \sigma_m \mod P
\right\}$ with $X^{(p)} > X^{(q)}$ for $p > q$.

Since $\boldsymbol{\xi}$ is a critical point,
$\partial \mathcal{C}_{\IIind}(\boldsymbol{\xi})/\partial X^{(p)} = 0$ (this is a necessary condition
of criticality but not sufficient). Hence, $\mathcal{C}_{\IIind}$ is invariant with
respect to contracting $\boldsymbol{\xi}$ to $\boldsymbol{\sigma}$:
$\mathcal{C}_{\IIind}(\boldsymbol{\xi}) = \mathcal{C}_{\IIind}(\boldsymbol{\sigma}, \mathbf{X}) =
\mathcal{C}_{\IIind}(\boldsymbol{\sigma}, \lambda\mathbf{X})$ for
$0 < \lambda < 1$. Thus, $\boldsymbol{\sigma} =: \boldsymbol{\sigma}(0) \in \sset{S}(c)$.

For $0 < r < r_1$, where $r_1 = 1 + X^{(1)}$, one has
$\boldsymbol{\sigma}(r) = \boldsymbol{\sigma}(0)$ and
$X^{(p)}(r) = X^{(p)} - r$. At $r = r_1$, spins at nodes in
$\sset{V}^{(1)}$ are inverted: $\sigma_m(r_1) = -\sigma_m(0)$ for all
$m \in \sset{V}^{(1)}$. This changes the rounded state, but the cut produced
by the new state is the same.

Indeed, expanding $\partial \mathcal{C}_{\IIind}(\boldsymbol{\xi})/\partial
X^{(1)} = 0$ (see Eq.~\eqref{eq:mdII-partition-xp}) produces
\begin{equation}\label{eq:mdII-eq-v1}
- \sum_{m \in \sset{V}^{(1)}} \sum_{n \notin \sset{V}^{(1)}} A_{m,n} \sigma_m \sigma_n = 0.
\end{equation}
This relation is invariant with respect to transformation
$\sigma_m \to -\sigma_m$ for all $m \in \sset{V}^{(1)}$. Thus, inverting spins in
$\sset{V}^{(1)}$ does not change the total cut.

For $r_1 < r < r_2 := 1 + X^{(2)}$, we have $\boldsymbol{\sigma}(r) =
\boldsymbol{\sigma}(r_1)$ and $X^{(2)}(r) < \ldots < X^{(S)}(r) < X^{(1)}(r)$. At $r
= r_2$, spins in $\sset{V}^{(2)}$ are inversed. The same argument as above
shows that this inversion does not change the cut size.

This process, increasing $r$ and inverting spins in $\sset{V}^{(i)}$ at
$r = r_i := 1 + X^{(i)}$, continues until $r$ reaches $r_S$. At this point,
one obtains $\boldsymbol{\sigma}(r_S) = - \boldsymbol{\sigma}(0)$ and the same mutual
relations between $X^{(p)} \left(r_S \right)$ as between
$X^{(p)} \left(0\right)$. Thus, while increasing $r$ (not necessarily
restricted to the interval $[0,2]$), rounding
$\widehat{R}_r[\boldsymbol{\xi}]$ produces the periodic sequence of binary
states
\begin{equation*}
  \left\{ \boldsymbol{\sigma}(0), \boldsymbol{\sigma}(r_1), \ldots,
  \boldsymbol{\sigma}(r_{S-1}), - \boldsymbol{\sigma}(0), -\boldsymbol{\sigma}(r_1), \ldots,
  -\boldsymbol{\sigma}(r_{S-1}), \boldsymbol{\sigma}(0), \ldots \right\}.
\end{equation*}
All these
states define cuts of the same size and, therefore, they are in $\sset{S}(c)$.

To complete the proof, one needs to show that one does not need to increase
the rounding center gradually and that $\widehat{R}_{r}[\boldsymbol{\xi}] =
\boldsymbol{\sigma}(r_i)$ for $r_i \leq r < r_{i+1}$. Indeed, in this case,
$\boldsymbol{\sigma}(r)$ is obtained from $\boldsymbol{\sigma}(0)$ by inverting spins
in $\sset{V}^{(1,i)} = \bigcup_{p \leq i} \sset{V}^{(i)}$, which yields
$\boldsymbol{\sigma}(r_i)$ since $\sset{V}^{(p)}$ are mutually disjoint.
\end{proof}

It follows from Theorem~\ref{thm:r-invariant-rounding} that for a critical
point $\boldsymbol{\xi}$ of $\mathcal{C}_{\IIind}$, the notion of cut size is correctly
defined even if $\boldsymbol{\xi}$ is not binary:
$\mathcal{C}(\boldsymbol{\xi}) := \mathcal{C} \left( \widehat{R}[\boldsymbol{\xi}] \right)$, where
$\widehat{R}[\boldsymbol{\xi}]$ is an arbitrary rounding.

The proof of Theorem~\ref{thm:r-invariant-rounding} shows that a
terminal state of the \mdII{} model with $S$ clusters generates $S$
different binary configurations corresponding to the same critical value of
$\mathcal{C}_{\IIind}(\boldsymbol{\sigma}, \mathbf{X})$. Thus, we obtain an important result
\begin{corollary} \label{cor:num-clusters} The number of clusters in a
  terminal state of the \mdII{} model $\boldsymbol{\xi}(\infty)$ does not exceed
  the number of binary states in the critical manifold $\sset{M}(c)$
  containing $\boldsymbol{\xi}(\infty)$.
\end{corollary}

The internal structure of terminal states revealed in the proof of
Theorem~\ref{thm:r-invariant-rounding} is illustrated by the results of
evolution of the \mdII{} model on $\sset{K}_{12}$ shown in
Fig.~\ref{fig:clusters-dyn}. The terminal state in this case has the form
of a family of clusters made of pairs of spins with opposite signs. It can
be shown that such a form is typical for terminal states fo the \mdII{}
model on complete graphs with even number of nodes starting from random
initial states. Inverting spins in the clusters starting from those with
the smallest $X$-component produces the series of different spin
configurations yielding the same value of cut, in this case, $36$, maximum
cut of $\sset{K}_{12}$. From the perspective of partitioning the graph
nodes, each new configuration corresponds to swapping two nodes from
different partitions. Obviously, for a maximum cut partition of a complete
graph, such swap does not change the number of cut edges. Notably, the
sequence of such swaps, or the family of multiple solutions, is obtained in
a single run of the \mdII-machine.



We conclude by noting that Theorem~\ref{thm:r-invariant-rounding} is the most
sensitive to regularizations of $\mathcal{C}_{\IIind}(\boldsymbol{\xi})$ as mapping
$\widehat{R}_{r}[\boldsymbol{\xi}]$ may not be well-defined for $r = X_{p}$.
In this case, for the regularized objective function,
Theorem~\ref{thm:r-invariant-rounding} holds for almost all points in
critical manifold $\sset{M}(c)$.

\subsection{Non-decreasing cuts and optimal rounding}

As discussed above, rounding of non-critical states produces cuts
of size depending on the choice of the rounding center. For
instance, this is the typical situation for equilibrium states of machines
based on rank-2 SDP. Therefore, the problem of recovering the best rounding
of such states needs to be specifically addressed. The main result of our
theoretical analysis of the \mdII{} model in the present paper is that the
\mdII-machine, by design, delivers such rounding. This follows from the observation
that dynamics $\boldsymbol{\xi}(t)$ does not depend on the choice of the
rounding center and an essential property of the \mdII-machine that
(small) perturbations of binary states cannot reduce the cut.

\begin{theorem} \label{thm:mdii-non-worsening} Let a binary state
  $\boldsymbol{\xi} = \boldsymbol{\sigma}$ be displaced by
  $\mathbf{X} \in (-1, 1]^N$, then the terminal machine's state yields cut of
  at least the same size as in the initial state:
  \begin{equation}\label{eq:mdII-non-worsening}
    \mathcal{C}\left(
      \boldsymbol{\xi}\left( \infty \vert \boldsymbol{\sigma} + \mathbf{X} \right)  \right)
    \geq \mathcal{C}(\boldsymbol{\sigma}).
  \end{equation}
\end{theorem}
\begin{proof}
  There are two mutually excluding scenarios of how the machine may evolve.

  The first scenario is when none of $X_m$, $m = 1, \ldots, N$, crosses
  $-1$ or $1$, so that $\boldsymbol{\sigma}(t)= \boldsymbol{\sigma}$. Then, by virtue
  of Theorem~\ref{thm:r-invariant-rounding}, the terminal state has the
  same cut as $\boldsymbol{\sigma}$.

  The second scenario occurs when one of $\left\{ X_m \right\}$, say,
  $X_p$, passes through $-1$ or $1$ leading to inverting the spin,
  $\sigma_p \to -\sigma_p$. As will be shown below, the new state has the cut strictly
  larger than $\mathcal{C}(\boldsymbol{\sigma})$. For the new state, we again have
  dynamics of a perturbed binary state. This dynamics also proceeds
  according to one of the two scenarios and so on. Since the total
  variation of cut is finite, $\boldsymbol{\sigma}$ may change only a finite
  number of times. Thus, the evolution arrives at the terminal state
  without decreasing the cut size.

  To complete the proof, we must show that the binary component may
  change only by increasing the cut. We consider the case when the
  variation occurs at crossing $-1$. Let $X_p$ be the component, which is
  about to cross $-1$, that is $X_p = \min_m X_m$ and $\dot{X}_p < 0$.
  Expanding Eq.~\eqref{eq:mdII-eqmo-sx}, we obtain
  \begin{equation}\label{eq:mdII-marginal-rate}
    \dot{X}_p = - \frac{1}{2} \sum_{n} A_{p,n} \sigma_p \sigma_n.
  \end{equation}
  Hence, $\dot{X}_p < 0$ only if $F(p) > 0$ (see Eq.~\eqref{eq:NMR}) and,
  therefore, inverting $\sigma_p$ increases the cut. The same conclusion holds
  when crossing occurs at $1$, or when a group of spins characterized by
  $X_m = X^{(p)}$ crosses $-1$ or $1$. It should be noted that while the
  condition $F_p > 0$ has locally the same form as for the greedy search,
  the structure of the terminal states of the \mdII-machine is distinctly
  different. The reason is the ability of multiple spins to invert
  their signs in the form of the cluster crossing the $\pm1$ boundary. The
  underlying reason is that the terminal states are determined by the
  \emph{global} structure of $\mathbf{X}$ (see
  Theorem~\ref{thm:prob_improve}).
\end{proof}

This property, weak perturbations do not decrease (by the time when the
terminal state is reached) the cut, is a distinguishing feature of the
\mdII-machine. For example, the machine based on rank-2 SDP does not have
this property. For such a machine, perturbing a binary state may lead to a
reduced cut.

Finally, applying these results to the progression of a non-critical state,
we obtain our main result. Starting from an arbitrary state, the
\mdII-machine terminates in a state with cut, which is not worse than
produced by the best rounding of the initial state.

\begin{theorem} \label{thm:optimal-rounding} Let the machine be initially
  in a non-critical state $\boldsymbol{\xi}(0)$ with
  $\overline{\mathcal{C}}\left( \boldsymbol{\xi}(0) \right) = \max_r \mathcal{C} \left(
    \widehat{R}_r[\boldsymbol{\xi}] \right)$ being the maximum cut that can
  be obtained by rounding it, and let
  $\boldsymbol{\xi}\left( \infty \vert \boldsymbol{\xi}^{(0)} \right)$ be the machine
  terminal state, then
  \begin{equation}\label{eq:dynamics-optimal-rounding}
    \mathcal{C} \left( \boldsymbol{\xi}\left( \infty \vert \boldsymbol{\xi}(0) \right) \right)
            \geq  \overline{C}\left( \boldsymbol{\xi}(0) \right).
  \end{equation}
\end{theorem}
\begin{proof}
  Since choosing the rounding center does not change
  $\mathcal{C}_{\IIind}\left(\boldsymbol{\xi}(t)\right)$, we can consider the dynamics
  as emerging from the perturbation of the best rounding state. Then, by
  virtue of Theorem~\ref{thm:mdii-non-worsening}, the cut cannot decrease.
\end{proof}

\section{Computational performance of the \mdII-machine}
\label{sec:numerics}

The computational effort of the Ising machine driven by the \mdII{} model
is represented by
terminal states of the \mdII{} model: $\boldsymbol{\xi}\left( \infty \vert
  \boldsymbol{\xi}(0) \right)$, where $\boldsymbol{\xi}(0)$ is the initial
state of the \mdII-machine.
A proper investigation of the evolution of probability distributions on the
phase space 
requires developing special approaches,
which are beyond the scope of the present paper. Therefore, we limit
ourselves to discussing only basic computational capabilities of the model
and their empirical demonstration.

In numerical experiments, the machine state is described by variables
$(\boldsymbol{\sigma}, \mathbf{X})$ with the update rule implemented using
procedure \textproc{Update} listed below. The dynamics of the \mdII{} model
was simulated using the Euler approximation. This approximation is exact
outside of
the discontinuities of the dynamical equations~\eqref{eq:mdII-eq-motion}
or~\eqref{eq:mdII-eqmo-sx}. On the other hand, since the magnitude of
$\dot{\boldsymbol{\xi}}$ does not depend on the proximity to equilibrium, the
Euler approximation with a fixed time step demonstrates spurious
oscillations near the critical point as characteristic to discontinuous
dynamical systems (see, e.g.~\cite{guglielmiEfficient2022}). Since reaching
equilibrium is important to ensure that main theorems in the previous
section
hold, these oscillations are expected to negatively impact the performance
of the \mdII-machine. On the other hand, the Euler approximation
effectively selects a regularization of
$\mathcal{C}_{\IIind}(\boldsymbol{\xi})$, which alleviates difficulties associated with
discontinuities of the dynamical equations of motion. From this
perspective, the numerical results presented below demonstrate the
robustness of the favorable features of the \mdII{} model with respect to
realizations of discontinuous dynamics.

\begin{algorithmic}[0]
  \Procedure{Update}{input $\left( \boldsymbol{\sigma}, \mathbf{X} \right)$,
    $\Delta\mathbf{X}$; output $\left( \boldsymbol{\sigma}, \mathbf{X} \right)$}
  \Ensure $\norm{\Delta \mathbf{X}}_\infty < 2$
  \State $\mathbf{X} \gets \mathbf{X} + \Delta\mathbf{X}$
  \ForAll{$m \in \left\{ 1, \dots, N \right\}$ }
  \If{$X_m > 1$}
     \State $X_m \gets X_m -2$
     \State $\sigma_m \gets -\sigma_m$
     \ElsIf{$X_m \leq -1$}
     \State $X_m \gets X_m + 2$
     \State $\sigma_m \gets -\sigma_m$
     \EndIf
  \EndFor
  \EndProcedure
\end{algorithmic}

The main features determining the computational capabilities of the \mdII{}
model are expressed by Theorems~\ref{thm:max_min_cut}
and~\ref{thm:mdii-non-worsening}. Starting from a random state
$(\boldsymbol{\sigma}, \mathbf{X})$
the evolution governed by Eq.~\eqref{eq:mdII-eqmo-sx}
is only guaranteed to terminate in a state, which is the optimally rounded
initial state. However, unless this state is happen to yield the maximum
cut, the state is not stable with respect to random perturbations.
Consequently, random agitations of binary states obtained from the terminal
state of the \mdII-machine with positive probability may lead to
improvement of the solution obtained by the machine.

\begin{theorem}\label{thm:prob_improve}
  Let $\boldsymbol{\sigma}$ be a non max-cut binary state,
  $\mathcal{C}(\boldsymbol{\sigma}) < \overline{C}_{\sset{G}}$. Let
  $\sset{A}(\boldsymbol{\sigma}) = \left\{ \mathbf{X} \in (-1, 1)^N:
    \widetilde{C}_{\IIind}(\boldsymbol{\sigma}, \mathbf{X}) > 0 \right\}$ be the
  set of agitations leading to an improved solution:
  $\mathcal{C}\left( \boldsymbol{\xi}\left( \infty \vert \boldsymbol{\sigma} + \mathbf{X} \right)
  \right) > \mathcal{C}(\boldsymbol{\sigma})$ for all
  $\mathbf{X} \in \sset{A}(\boldsymbol{\sigma})$. Then
  $\abs{\sset{A}(\boldsymbol{\sigma})} > 0$, that is $\sset{A}(\boldsymbol{\sigma})$
  has finite volume in $(-1, 1)^N$.
\end{theorem}
\begin{proof}
  Set $\sset{A}(\boldsymbol{\sigma})$ is nonempty, as follows from
  Theorem~\ref{thm:max_min_cut}. In turn, the dynamics governed by
  Eqs.~\eqref{eq:mdII-eqmo-sx} results in non-decreasing cut. Then, the
  theorem's statement follows from set $\sset{A}(\boldsymbol{\sigma})$ being
  open: with any point $\mathbf{X}$, it contains a ball of finite radius in
  $(-1, 1)^N$ centered at $\mathbf{X}$, which, in turn, follows from
  $\widetilde{C}_{\IIind}(\boldsymbol{\sigma}, \mathbf{X})$ being continuous
  function of $\mathbf{X} \in (-1, 1)^N$.
\end{proof}

\begin{figure}[tb]
  \centering
  \includegraphics[width=3in]{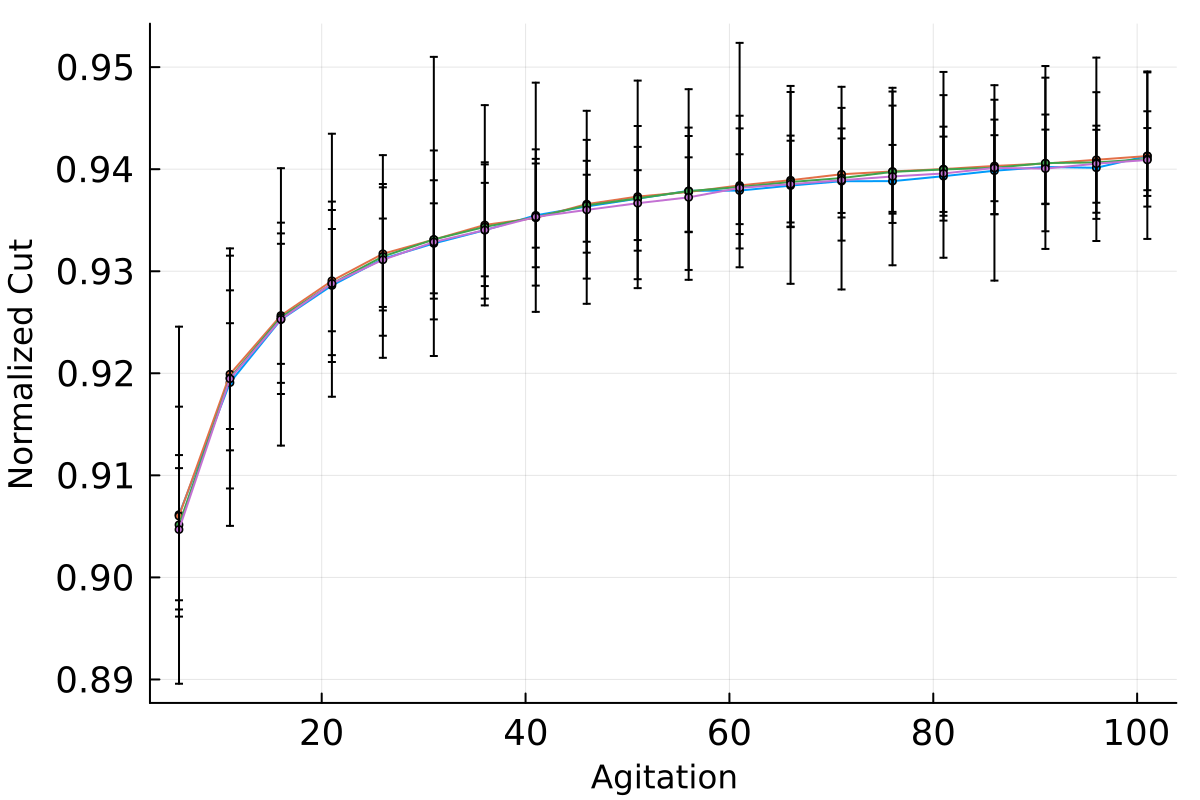}
  \caption{The dependence of the normalized cut,
    $c_D = \left( \mathcal{C} / M - 1 /2 \right) \sqrt{D} / P_*$, obtained by the
    \mdII-machine on the number of agitations for four random
    $(D = 3)$-regular graphs with $6,400$, $12,800$, $25,000$ and $40,000$
    nodes. For each number of agitations,
    $c_D\left( \boldsymbol{\sigma}^{(h)} \right)$ was evaluated for $100$ random
    initial states. The solid line and the
    error bars show the mean, and the maximal and minimal values of
    $c_D\left( \boldsymbol{\sigma}^{(h)} \right)$. }
  \label{fig:agg_progression}
\end{figure}

To employ this property, we introduce an agitated progression of the
\mdII-machine: when the machine reached the terminal state (as was shown in
Section~\ref{sec:basic-terminals}, the time required for reaching the terminal
state is finite), the continuous component is discarded and reinitialized
by choosing $\mathbf{X}$ randomly from $(-1, 1)^N$. Starting from an
arbitrary binary state $\boldsymbol{\sigma}^{(0)}$, this procedure results in a
sequence of binary states
$\boldsymbol{\sigma}^{(h)} = \widehat{R}_{r}\left[ \boldsymbol{\xi}\left( \infty \vert
    \boldsymbol{\sigma}^{(h-1)} + \mathbf{X}^{(h)} \right) \right]$, where
$h$ is the number of agitations and $\mathbf{X}^{(h)} \in (-1, 1)^N$ is the
agitation vector. By virtue of Theorem~\ref{thm:mdii-non-worsening},
$\mathcal{C}\left( \boldsymbol{\sigma}^{(h+1)} \right) \geq \mathcal{C}\left( \boldsymbol{\sigma}^{(h)}
\right)$, and from Theorem~\ref{thm:prob_improve} immediately follows
\begin{corollary}
  Let $X_m^{(h)}$, the components of the agitation vector
  $\mathbf{X}^{(h)}$, be uniformly distributed in $(-1, 1)$, then the
  agitated progression converges to a maximum-cut binary state
  \begin{equation}\label{eq:agit_converge}
    \lim_{h \to \infty} \mathcal{C}\left( \boldsymbol{\sigma}^{(h)} \right) = \overline{C}_{\sset{G}},
  \end{equation}
  with probability $1$.
\end{corollary}

Figure~\ref{fig:agg_progression} illustrates the progression of the
\mdII-machine towards the maximum cut with the number of agitations for the
example of four random $3$-regular graphs with $6,400$, $12,800$, $25,000$
and $40,000$ nodes. Regular graphs were chosen because the normalized
maximum cut of $D$-regular random graphs has a well-established asymptotic
with the number of nodes:
$\overline{c}_D = \left( \frac{\overline{\mathcal{C}}_{\sset{G}}}{M_{\sset{G}}} -
  \frac{1}{2} \right) \frac{\sqrt{D}}{P_*} \sim 1$, where
$M_{\sset{G}} = D N_{\sset{G}} / 2$ is the number of graph edges and
$P_* \approx 0.763166$ is the Parisi constant~\cite{demboExtremal2017}. To
emphasize the high values of the normalized cut obtained by the \mdII{}
model, we note that the SDP relaxation yields
$c_D^{(SDP)} = 2 /(\pi P_*) \approx 0.834$~\cite{fanHow2017}.

The property to obtain
improved solutions by means of the agitated progression endows the
\mdII-model with its own non-trivial computational capabilities.
Figure~\ref{fig:cmp_scaling} demonstrates these capabilities by comparing
the \mdII-machine on a set of random $3$-regular graphs with $1$-opt local
search (greedy search ensuring that for each node at least half of the
incident edges are cut) and a software simulation of the coherent Ising
machine~\cite{yamamotoCoherent2017}. The coherent Ising machine
(CIM)~\cite{yamamotoCoherent2017} was simulated using package
\verb|cim_optimizer|~\cite{Chen_cim-optimizer_a_simulator_2022}. 

\begin{figure}[tb]
  \centering
  \includegraphics[width=3.5in]{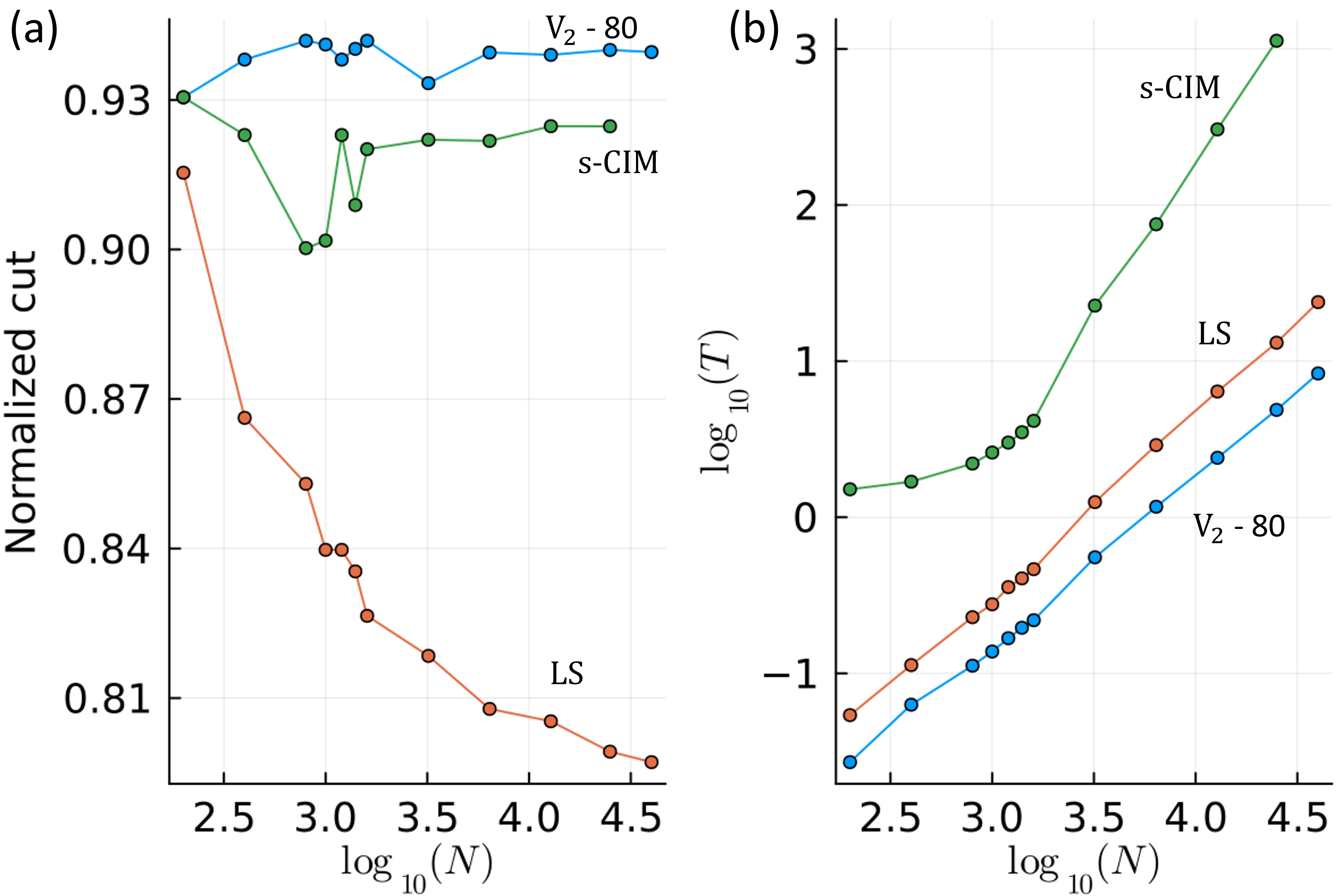}
  \caption{Comparison of the performance of the \mdII-machine ($h=80$
    agitations) with $1$-opt local search (the best value out of $8000$
    restarts) and the simulated coherent Ising
    machine on a set of random $3$-regular graphs. (a) The normalized value
    of cut as a function of the number of nodes. (b)
    Scaling of wall-time with the number of nodes.}
  \label{fig:cmp_scaling}
\end{figure}

The increasing discrepancy between the local search and the
\mdII{} and the s-$\mathrm{CIM}$ lines in Fig.~\ref{fig:cmp_scaling}(a)
clearly demonstrates that both the \mdII{} machine and the coherent Ising
machine converge to their terminal states following mechanisms that are
distinct from that of the greedy search algorithm. This distinction is
especially important for the \mdII-machine as the condition for changing
the binary component [$F_p > 0$ with $F_p$ defined in Eq.~\eqref{eq:NMR}]
has locally the same form as that of the greedy search, which may create an
impression of a tight relation between the final states obtained by the
\mdII-machine and the greedy search.

\begin{figure}[tb]
  \centering
  \includegraphics[width=3.1in]{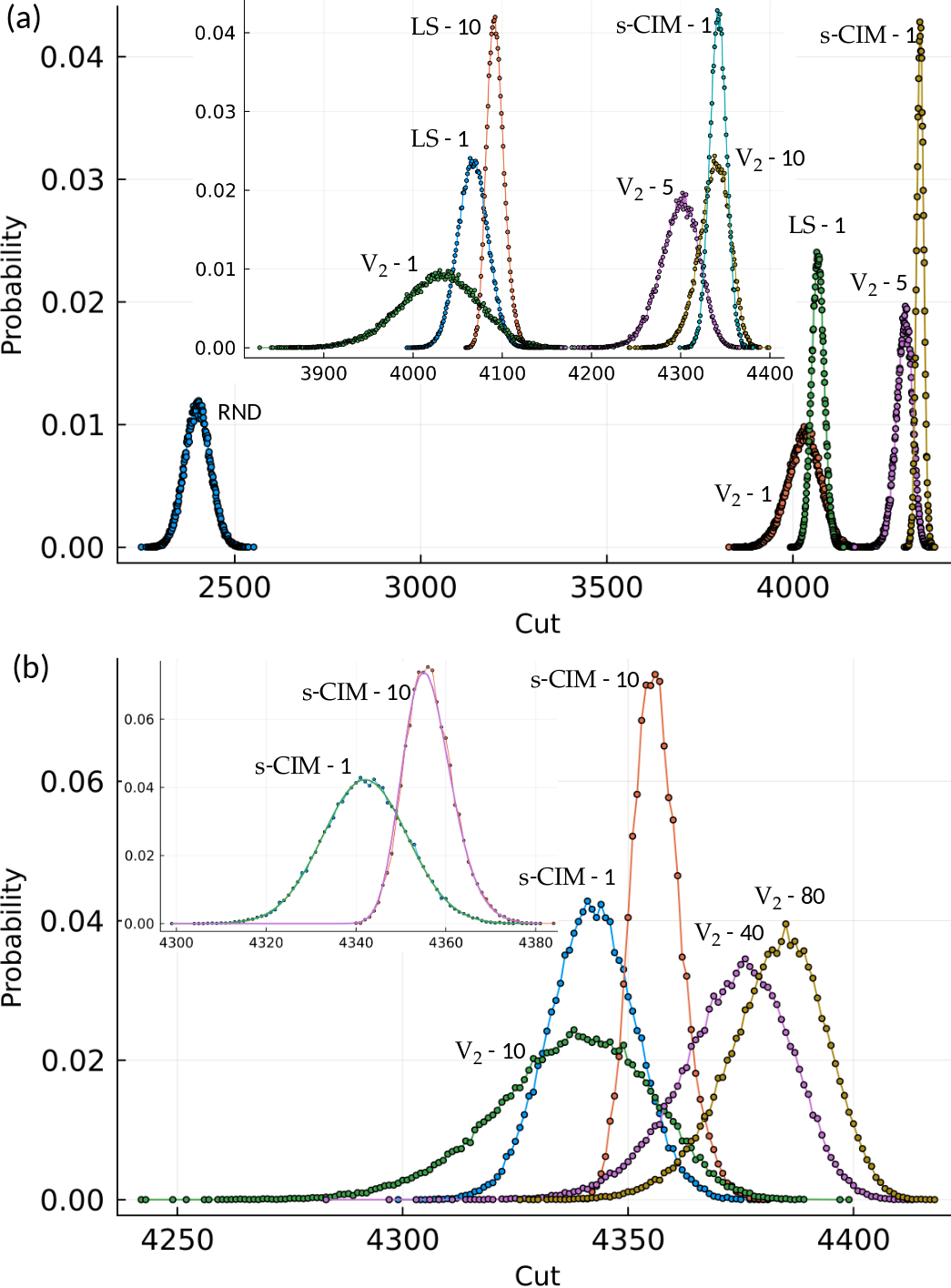}
  \caption{Distributions of cuts obtained by the \mdII-machine (lines
    $\mathrm{V}_2$-$h$, where $h$ is the number of agitations), $1$-opt
    local search and simulated coherent Ising machine (lines
    $\mathrm{LS}$-$p$ and s-$\mathrm{CIM}$-$p$, where $p$ denotes the best
    out of $p$ runs). (a) Overall comparison of the distributions of cuts
    produced by random partitioning (RND), local search, and the
    \mdII-machine. The inset shows the details of the distribution
    functions on intermediate values of obtained cut. (b) The details of
    the distribution functions at high cut values. The inset shows fitting
    the s-CIM's empirical distribution functions by the normal
    distribution.}
  \label{fig:distr-cmp}
\end{figure}

Figure~\ref{fig:cmp_scaling} does not reflect the probabilistic nature
of the compared approaches. To get a better insight,
Fig.~\ref{fig:distr-cmp}(a,b) show empirical probability distribution
functions obtained by sampling the results obtained starting from at least
$6 \cdot 10^4$ random initial conditions for a random $3$-regular graph with
$3,200$ nodes. These results demonstrate that the improvement of the
quality of solutions obtained by the \mdII{} model with the number of
agitations has a different origin than mere reiteration of the machine.
This observation is emphasized by Fig.~\ref{fig:distr-cmp}(b) depicting the
empirical distribution functions at high values of obtained cuts. Using the
observation that the distribution of the outcomes of the coherent Ising
machine is well-approximated by the Gaussian function (with the mean value
and variance equal to $\mu_{\mathrm{CIM}} = 4342$ and
$v_{\mathrm{CIM}} = 9.5$, respectively), we can easily estimate the number
of reiteration of the simulated coherent Ising machine to match the maximum
of the distribution of the \mdII{} results. The cumulative distribution
function of the best out of $p$ runs is
$P_p(\mathcal{C} < c) = P_1^p(\mathcal{C} < c)$ with the respective distribution function
$\rho_p(\mathcal{C} = c) = p \rho_1 (\mathcal{C} = c) P_1^{p-1}(\mathcal{C} < c)$. The inset in
Fig.~\ref{fig:distr-cmp}(b) shows that the best-out-of-$p$ distribution
derived from the fitted Gaussian function approximates well the empirical
distribution function. Then, the number of reruns needed to match the most
probable outcome of the coherent Ising machine with the mean value of the
agitated \mdII-machine outcomes (it does not exceed the most probable
value) is
\begin{equation}\label{eq:matching_cim_v2}
  k_h = 1 + \sqrt{\pi} \left[ 1 + \erf(d_h) \right] d_h e^{d_h^2},
\end{equation}
where 
$d_h = \left( \mu_{h} - \mu_{\mathrm{CIM}} \right) / v_{\mathrm{CIM}}
\sqrt{2}$, and $\mu_h$ is the mean value of the cut produced by the
\mdII-machine after $h$ agitations. For the graph used in
Fig.~\ref{fig:distr-cmp}, we have $\mu_{20} = 4359$, $\mu_{40} = 4373$, and
$\mu_{80} = 4382$, which yields $k_{20} \approx 22$, $k_{40} \approx 1.7\cdot10^3$, $k_{80} \approx
7.5 \cdot 10^{4}$, respectively. 


\section{Conclusion}

The connection between the ground state of the classical spin system and
the max-cut problem provides means for evaluating and projecting the
computational capabilities of Ising machines, for instance, how their
performance will scale with the problem size. In recent
papers~\cite{erementchoukComputational2022, shuklaScalable2022,
  shuklaCustom2023}, we approached the problem of designing Ising machines
from the perspective of employing dynamical models inspired by
relaxation-based techniques, such as the SDP relaxation. These techniques
are known to demonstrate favorable scaling properties. However, the
respective dynamical systems settle in a non-binary state, and the
final state must be rounded to recover a feasible solution to the Ising or
max-cut problem. The known rounding algorithms, even in the case of rank-2
relaxations, associating spins with planar unit vectors, require
external processing power. This makes the information processing flow in
relaxation-based Ising machines incomplete.

We show that a special dynamical model (we call it the \mdII{} model)
possesses the key property: given a non-binary initial state, it evolves
towards a trivially rounding state, which yields the cut that, at least, is
not smaller than obtained by the best rounding of the initial state.

Another important property of the \mdII{} model is tightly related to the
ability to deliver the best rounding. We show that if the \mdII{} model
evolves from a slight perturbation of a binary state, it ends up in a state
producing a cut not smaller than that of the original binary state, which
enables improving the solution quality. We show that such an agitated
machine converges to a maximum cut state almost surely.

We demonstrate favorable computational capabilities of the \mdII-machine by
comparing it numerically with the $1$-opt local search and coherent Ising
machine on random $3$-regular graphs.

Thus, incorporating the \mdII{} model as the final stage of a heterogeneous
Ising machine eliminates the necessity for the external processing of
relaxation-based Ising machines and makes them self-contained.
Consequently, any dynamical system with the phase space consistent with the
\mdII{} model can be used to drive the Ising machine and can be chosen
solely on the ground of computational performance on particular instances
of optimization problems.

\section*{Acknowledgements}
The work was partially supported by the US National Science Foundation (NSF) under
Grant No. 1909937 and Grant No. 2531175.


\begin{thebibliography}{51}
\ifx \bisbn   \undefined \def \bisbn  #1{ISBN #1}\fi
\ifx \binits  \undefined \def \binits#1{#1}\fi
\ifx \bauthor  \undefined \def \bauthor#1{#1}\fi
\ifx \batitle  \undefined \def \batitle#1{#1}\fi
\ifx \bjtitle  \undefined \def \bjtitle#1{#1}\fi
\ifx \bvolume  \undefined \def \bvolume#1{\textbf{#1}}\fi
\ifx \byear  \undefined \def \byear#1{#1}\fi
\ifx \bissue  \undefined \def \bissue#1{#1}\fi
\ifx \bfpage  \undefined \def \bfpage#1{#1}\fi
\ifx \blpage  \undefined \def \blpage #1{#1}\fi
\ifx \burl  \undefined \def \burl#1{\textsf{#1}}\fi
\ifx \doiurl  \undefined \def \doiurl#1{\url{https://doi.org/#1}}\fi
\ifx \betal  \undefined \def \betal{\textit{et al.}}\fi
\ifx \binstitute  \undefined \def \binstitute#1{#1}\fi
\ifx \binstitutionaled  \undefined \def \binstitutionaled#1{#1}\fi
\ifx \bctitle  \undefined \def \bctitle#1{#1}\fi
\ifx \beditor  \undefined \def \beditor#1{#1}\fi
\ifx \bpublisher  \undefined \def \bpublisher#1{#1}\fi
\ifx \bbtitle  \undefined \def \bbtitle#1{#1}\fi
\ifx \bedition  \undefined \def \bedition#1{#1}\fi
\ifx \bseriesno  \undefined \def \bseriesno#1{#1}\fi
\ifx \blocation  \undefined \def \blocation#1{#1}\fi
\ifx \bsertitle  \undefined \def \bsertitle#1{#1}\fi
\ifx \bsnm \undefined \def \bsnm#1{#1}\fi
\ifx \bsuffix \undefined \def \bsuffix#1{#1}\fi
\ifx \bparticle \undefined \def \bparticle#1{#1}\fi
\ifx \barticle \undefined \def \barticle#1{#1}\fi
\bibcommenthead
\ifx \bconfdate \undefined \def \bconfdate #1{#1}\fi
\ifx \botherref \undefined \def \botherref #1{#1}\fi
\ifx \url \undefined \def \url#1{\textsf{#1}}\fi
\ifx \bchapter \undefined \def \bchapter#1{#1}\fi
\ifx \bbook \undefined \def \bbook#1{#1}\fi
\ifx \bcomment \undefined \def \bcomment#1{#1}\fi
\ifx \oauthor \undefined \def \oauthor#1{#1}\fi
\ifx \citeauthoryear \undefined \def \citeauthoryear#1{#1}\fi
\ifx \endbibitem  \undefined \def \endbibitem {}\fi
\ifx \bconflocation  \undefined \def \bconflocation#1{#1}\fi
\ifx \arxivurl  \undefined \def \arxivurl#1{\textsf{#1}}\fi
\csname PreBibitemsHook\endcsname

\bibitem[\protect\citeauthoryear{Kirkpatrick
  et~al.}{1983}]{kirkpatrickOptimization1983}
\begin{barticle}
\bauthor{\bsnm{Kirkpatrick}, \binits{S.}},
\bauthor{\bsnm{Gelatt}, \binits{C.D.}},
\bauthor{\bsnm{Vecchi}, \binits{M.P.}}:
\batitle{Optimization by {{Simulated Annealing}}}.
\bjtitle{Science}
\bvolume{220}(\bissue{4598}),
\bfpage{671}--\blpage{680}
(\byear{1983})
\doiurl{10.1126/science.220.4598.671}
\end{barticle}
\endbibitem

\bibitem[\protect\citeauthoryear{Hopfield}{1984}]{hopfieldNeurons1984}
\begin{barticle}
\bauthor{\bsnm{Hopfield}, \binits{J.J.}}:
\batitle{Neurons with graded response have collective computational properties
  like those of two-state neurons.}
\bjtitle{Proceedings of the National Academy of Sciences}
\bvolume{81}(\bissue{10}),
\bfpage{3088}--\blpage{3092}
(\byear{1984})
\doiurl{10.1073/pnas.81.10.3088}
\end{barticle}
\endbibitem

\bibitem[\protect\citeauthoryear{{\v
  C}ern{\'y}}{1985}]{cernyThermodynamical1985}
\begin{barticle}
\bauthor{\bsnm{{\v C}ern{\'y}}, \binits{V.}}:
\batitle{Thermodynamical approach to the traveling salesman problem: {{An}}
  efficient simulation algorithm}.
\bjtitle{Journal of Optimization Theory and Applications}
\bvolume{45}(\bissue{1}),
\bfpage{41}--\blpage{51}
(\byear{1985})
\doiurl{10.1007/BF00940812}
\end{barticle}
\endbibitem

\bibitem[\protect\citeauthoryear{Fu and Anderson}{1986}]{fuApplication1986}
\begin{barticle}
\bauthor{\bsnm{Fu}, \binits{Y.}},
\bauthor{\bsnm{Anderson}, \binits{P.W.}}:
\batitle{Application of statistical mechanics to {{NP-complete}} problems in
  combinatorial optimisation}.
\bjtitle{Journal of Physics A: Mathematical and General}
\bvolume{19}(\bissue{9}),
\bfpage{1605}--\blpage{1620}
(\byear{1986})
\doiurl{10.1088/0305-4470/19/9/033}
\end{barticle}
\endbibitem

\bibitem[\protect\citeauthoryear{Kochenberger
  et~al.}{2013}]{kochenbergerBinary2013}
\begin{bchapter}
\bauthor{\bsnm{Kochenberger}, \binits{G.A.}},
\bauthor{\bsnm{Glover}, \binits{F.}},
\bauthor{\bsnm{Wang}, \binits{H.}}:
\bctitle{Binary {{Unconstrained Quadratic Optimization Problem}}}.
In: \beditor{\bsnm{Pardalos}, \binits{P.M.}},
\beditor{\bsnm{Du}, \binits{D.-Z.}},
\beditor{\bsnm{Graham}, \binits{R.L.}} (eds.)
\bbtitle{Handbook of {{Combinatorial Optimization}}},
pp. \bfpage{533}--\blpage{557}.
\bpublisher{{Springer}},
\blocation{{New York, NY}}
(\byear{2013}).
\doiurl{10.1007/978-1-4419-7997-1_15}
\end{bchapter}
\endbibitem

\bibitem[\protect\citeauthoryear{Barahona}{1982}]{barahonaComputational1982}
\begin{barticle}
\bauthor{\bsnm{Barahona}, \binits{F.}}:
\batitle{On the computational complexity of {{Ising}} spin glass models}.
\bjtitle{Journal of Physics A: Mathematical and General}
\bvolume{15}(\bissue{10}),
\bfpage{3241}--\blpage{3253}
(\byear{1982})
\doiurl{10.1088/0305-4470/15/10/028}
\end{barticle}
\endbibitem

\bibitem[\protect\citeauthoryear{Karp}{1972}]{karpReducibility1972}
\begin{bchapter}
\bauthor{\bsnm{Karp}, \binits{R.M.}}:
\bctitle{Reducibility among {{Combinatorial Problems}}}.
In: \beditor{\bsnm{Miller}, \binits{R.E.}},
\beditor{\bsnm{Thatcher}, \binits{J.W.}},
\beditor{\bsnm{Bohlinger}, \binits{J.D.}} (eds.)
\bbtitle{Complexity of {{Computer Computations}}},
pp. \bfpage{85}--\blpage{103}.
\bpublisher{{Springer US}},
\blocation{{Boston, MA}}
(\byear{1972}).
\doiurl{10.1007/978-1-4684-2001-2_9}
\end{bchapter}
\endbibitem

\bibitem[\protect\citeauthoryear{Garey et~al.}{1976}]{gareySimplified1976}
\begin{barticle}
\bauthor{\bsnm{Garey}, \binits{M.R.}},
\bauthor{\bsnm{Johnson}, \binits{D.S.}},
\bauthor{\bsnm{Stockmeyer}, \binits{L.}}:
\batitle{Some simplified {{NP-complete}} graph problems}.
\bjtitle{Theoretical Computer Science}
\bvolume{1}(\bissue{3}),
\bfpage{237}--\blpage{267}
(\byear{1976})
\doiurl{10.1016/0304-3975(76)90059-1}
\end{barticle}
\endbibitem

\bibitem[\protect\citeauthoryear{Lucas}{2014}]{lucasIsing2014}
\begin{barticle}
\bauthor{\bsnm{Lucas}, \binits{A.}}:
\batitle{Ising formulations of many {{NP}} problems}.
\bjtitle{Frontiers in Physics}
\bvolume{2},
\bfpage{5}
(\byear{2014})
\doiurl{10.3389/fphy.2014.00005}
\end{barticle}
\endbibitem

\bibitem[\protect\citeauthoryear{Aadit et~al.}{2022}]{aaditMassively2022}
\begin{barticle}
\bauthor{\bsnm{Aadit}, \binits{N.A.}},
\bauthor{\bsnm{Grimaldi}, \binits{A.}},
\bauthor{\bsnm{Carpentieri}, \binits{M.}},
\bauthor{\bsnm{Theogarajan}, \binits{L.}},
\bauthor{\bsnm{Martinis}, \binits{J.M.}},
\bauthor{\bsnm{Finocchio}, \binits{G.}},
\bauthor{\bsnm{Camsari}, \binits{K.Y.}}:
\batitle{Massively parallel probabilistic computing with sparse {{Ising}}
  machines}.
\bjtitle{Nature Electronics}
\bvolume{5}(\bissue{7}),
\bfpage{460}--\blpage{468}
(\byear{2022})
\doiurl{10.1038/s41928-022-00774-2}
\end{barticle}
\endbibitem

\bibitem[\protect\citeauthoryear{Tatsumura}{2021}]{tatsumuraLargescale2021}
\begin{bchapter}
\bauthor{\bsnm{Tatsumura}, \binits{K.}}:
\bctitle{Large-scale combinatorial optimization in real-time systems by
  {{FPGA-based}} accelerators for simulated bifurcation}.
In: \bbtitle{Proceedings of the 11th {{International Symposium}} on {{Highly
  Efficient Accelerators}} and {{Reconfigurable Technologies}}},
pp. \bfpage{1}--\blpage{6}.
\bpublisher{{ACM}},
\blocation{{Online Germany}}
(\byear{2021}).
\doiurl{10.1145/3468044.3468045}
\end{bchapter}
\endbibitem

\bibitem[\protect\citeauthoryear{Tatsumura et~al.}{2021}]{tatsumuraScaling2021}
\begin{barticle}
\bauthor{\bsnm{Tatsumura}, \binits{K.}},
\bauthor{\bsnm{Yamasaki}, \binits{M.}},
\bauthor{\bsnm{Goto}, \binits{H.}}:
\batitle{Scaling out {{Ising}} machines using a multi-chip architecture for
  simulated bifurcation}.
\bjtitle{Nature Electronics}
\bvolume{4}(\bissue{3}),
\bfpage{208}--\blpage{217}
(\byear{2021})
\doiurl{10.1038/s41928-021-00546-4}
\end{barticle}
\endbibitem

\bibitem[\protect\citeauthoryear{Patel et~al.}{2022}]{patelLogically2022}
\begin{barticle}
\bauthor{\bsnm{Patel}, \binits{S.}},
\bauthor{\bsnm{Canoza}, \binits{P.}},
\bauthor{\bsnm{Salahuddin}, \binits{S.}}:
\batitle{Logically synthesized and hardware-accelerated restricted
  {{Boltzmann}} machines for combinatorial optimization and integer
  factorization}.
\bjtitle{Nature Electronics}
\bvolume{5}(\bissue{2}),
\bfpage{92}--\blpage{101}
(\byear{2022})
\doiurl{10.1038/s41928-022-00714-0}
\end{barticle}
\endbibitem

\bibitem[\protect\citeauthoryear{Yamamoto et~al.}{2021}]{yamamotoSTATICA2021}
\begin{barticle}
\bauthor{\bsnm{Yamamoto}, \binits{K.}},
\bauthor{\bsnm{Kawamura}, \binits{K.}},
\bauthor{\bsnm{Ando}, \binits{K.}},
\bauthor{\bsnm{Mertig}, \binits{N.}},
\bauthor{\bsnm{Takemoto}, \binits{T.}},
\bauthor{\bsnm{Yamaoka}, \binits{M.}},
\bauthor{\bsnm{Teramoto}, \binits{H.}},
\bauthor{\bsnm{Sakai}, \binits{A.}},
\bauthor{\bsnm{{Takamaeda-Yamazaki}}, \binits{S.}},
\bauthor{\bsnm{Motomura}, \binits{M.}}:
\batitle{{{STATICA}}: {{A}} 512-{{Spin}} 0.{{25M-Weight Annealing Processor
  With}} an {{All-Spin-Updates-at-Once Architecture}} for {{Combinatorial
  Optimization With Complete Spin}}\textendash{{Spin Interactions}}}.
\bjtitle{IEEE Journal of Solid-State Circuits}
\bvolume{56}(\bissue{1}),
\bfpage{165}--\blpage{178}
(\byear{2021})
\doiurl{10.1109/JSSC.2020.3027702}
\end{barticle}
\endbibitem

\bibitem[\protect\citeauthoryear{Yamaoka et~al.}{2016}]{yamaoka20kSpin2016}
\begin{barticle}
\bauthor{\bsnm{Yamaoka}, \binits{M.}},
\bauthor{\bsnm{Yoshimura}, \binits{C.}},
\bauthor{\bsnm{Hayashi}, \binits{M.}},
\bauthor{\bsnm{Okuyama}, \binits{T.}},
\bauthor{\bsnm{Aoki}, \binits{H.}},
\bauthor{\bsnm{Mizuno}, \binits{H.}}:
\batitle{A 20k-{{Spin Ising Chip}} to {{Solve Combinatorial Optimization
  Problems With CMOS Annealing}}}.
\bjtitle{IEEE Journal of Solid-State Circuits}
\bvolume{51}(\bissue{1}),
\bfpage{303}--\blpage{309}
(\byear{2016})
\doiurl{10.1109/JSSC.2015.2498601}
\end{barticle}
\endbibitem

\bibitem[\protect\citeauthoryear{Ahmed et~al.}{2021}]{ahmedProbabilistic2021}
\begin{botherref}
\oauthor{\bsnm{Ahmed}, \binits{I.}},
\oauthor{\bsnm{Chiu}, \binits{P.-W.}},
\oauthor{\bsnm{Moy}, \binits{W.}},
\oauthor{\bsnm{Kim}, \binits{C.H.}}:
A {{Probabilistic Compute Fabric Based}} on {{Coupled Ring Oscillators}} for
  {{Solving Combinatorial Optimization Problems}}.
IEEE Journal of Solid-State Circuits,
1--1
(2021)
\doiurl{10.1109/JSSC.2021.3062821}
\end{botherref}
\endbibitem

\bibitem[\protect\citeauthoryear{Moy et~al.}{2022}]{moy968node2022}
\begin{barticle}
\bauthor{\bsnm{Moy}, \binits{W.}},
\bauthor{\bsnm{Ahmed}, \binits{I.}},
\bauthor{\bsnm{Chiu}, \binits{P.-w.}},
\bauthor{\bsnm{Moy}, \binits{J.}},
\bauthor{\bsnm{Sapatnekar}, \binits{S.S.}},
\bauthor{\bsnm{Kim}, \binits{C.H.}}:
\batitle{A 1,968-node coupled ring oscillator circuit for combinatorial
  optimization problem solving}.
\bjtitle{Nature Electronics}
\bvolume{5}(\bissue{5}),
\bfpage{310}--\blpage{317}
(\byear{2022})
\doiurl{10.1038/s41928-022-00749-3}
\end{barticle}
\endbibitem

\bibitem[\protect\citeauthoryear{Afoakwa et~al.}{2021}]{afoakwaBRIM2021}
\begin{bchapter}
\bauthor{\bsnm{Afoakwa}, \binits{R.}},
\bauthor{\bsnm{Zhang}, \binits{Y.}},
\bauthor{\bsnm{Vengalam}, \binits{U.K.R.}},
\bauthor{\bsnm{Ignjatovic}, \binits{Z.}},
\bauthor{\bsnm{Huang}, \binits{M.}}:
\bctitle{{{BRIM}}: {{Bistable Resistively-Coupled Ising Machine}}}.
In: \bbtitle{2021 {{IEEE International Symposium}} on {{High-Performance
  Computer Architecture}} ({{HPCA}})},
pp. \bfpage{749}--\blpage{760}.
\bpublisher{{IEEE}},
\blocation{{Seoul, Korea (South)}}
(\byear{2021}).
\doiurl{10.1109/HPCA51647.2021.00068}
\end{bchapter}
\endbibitem

\bibitem[\protect\citeauthoryear{Leleu et~al.}{2021}]{leleuScaling2021}
\begin{barticle}
\bauthor{\bsnm{Leleu}, \binits{T.}},
\bauthor{\bsnm{Khoyratee}, \binits{F.}},
\bauthor{\bsnm{Levi}, \binits{T.}},
\bauthor{\bsnm{Hamerly}, \binits{R.}},
\bauthor{\bsnm{Kohno}, \binits{T.}},
\bauthor{\bsnm{Aihara}, \binits{K.}}:
\batitle{Scaling advantage of chaotic amplitude control for high-performance
  combinatorial optimization}.
\bjtitle{Communications Physics}
\bvolume{4}(\bissue{1}),
\bfpage{266}
(\byear{2021})
\doiurl{10.1038/s42005-021-00768-0}
\end{barticle}
\endbibitem

\bibitem[\protect\citeauthoryear{Kuramoto}{1975}]{kuramotoSelfentrainment1975}
\begin{bchapter}
\bauthor{\bsnm{Kuramoto}, \binits{Y.}}:
\bctitle{Self-entrainment of a population of coupled non-linear oscillators}.
In: \beditor{\bsnm{Araki}, \binits{H.}} (ed.)
\bbtitle{International {{Symposium}} on {{Mathematical Problems}} in
  {{Theoretical Physics}}}
vol. \bseriesno{39},
pp. \bfpage{420}--\blpage{422}.
\bpublisher{{Springer-Verlag}},
\blocation{{Berlin/Heidelberg}}
(\byear{1975}).
\doiurl{10.1007/BFb0013365}
\end{bchapter}
\endbibitem

\bibitem[\protect\citeauthoryear{Shinomoto and
  Kuramoto}{1986}]{shinomotoPhase1986}
\begin{barticle}
\bauthor{\bsnm{Shinomoto}, \binits{S.}},
\bauthor{\bsnm{Kuramoto}, \binits{Y.}}:
\batitle{Phase {{Transitions}} in {{Active Rotator Systems}}}.
\bjtitle{Progress of Theoretical Physics}
\bvolume{75}(\bissue{5}),
\bfpage{1105}--\blpage{1110}
(\byear{1986})
\doiurl{10.1143/PTP.75.1105}
\end{barticle}
\endbibitem

\bibitem[\protect\citeauthoryear{Mori and Kuramoto}{1998}]{moriDissipative1998}
\begin{bbook}
\bauthor{\bsnm{Mori}, \binits{H.}},
\bauthor{\bsnm{Kuramoto}, \binits{Y.}}:
\bbtitle{Dissipative Structures and Chaos}.
\bpublisher{{Springer}},
\blocation{{Berlin ; New York}}
(\byear{1998})
\end{bbook}
\endbibitem

\bibitem[\protect\citeauthoryear{Acebr{\'o}n
  et~al.}{2005}]{acebronKuramoto2005}
\begin{barticle}
\bauthor{\bsnm{Acebr{\'o}n}, \binits{J.A.}},
\bauthor{\bsnm{Bonilla}, \binits{L.L.}},
\bauthor{\bsnm{P{\'e}rez~Vicente}, \binits{C.J.}},
\bauthor{\bsnm{Ritort}, \binits{F.}},
\bauthor{\bsnm{Spigler}, \binits{R.}}:
\batitle{The {{Kuramoto}} model: {{A}} simple paradigm for synchronization
  phenomena}.
\bjtitle{Reviews of Modern Physics}
\bvolume{77}(\bissue{1}),
\bfpage{137}--\blpage{185}
(\byear{2005})
\doiurl{10.1103/RevModPhys.77.137}
\end{barticle}
\endbibitem

\bibitem[\protect\citeauthoryear{Albertsson
  et~al.}{2021}]{albertssonUltrafast2021}
\begin{barticle}
\bauthor{\bsnm{Albertsson}, \binits{D.I.}},
\bauthor{\bsnm{Zahedinejad}, \binits{M.}},
\bauthor{\bsnm{Houshang}, \binits{A.}},
\bauthor{\bsnm{Khymyn}, \binits{R.}},
\bauthor{\bsnm{{\AA}kerman}, \binits{J.}},
\bauthor{\bsnm{Rusu}, \binits{A.}}:
\batitle{Ultrafast {{Ising Machines}} using spin torque nano-oscillators}.
\bjtitle{Applied Physics Letters}
\bvolume{118}(\bissue{11}),
\bfpage{112404}
(\byear{2021})
\doiurl{10.1063/5.0041575}
\end{barticle}
\endbibitem

\bibitem[\protect\citeauthoryear{Wang et~al.}{2021}]{wangSolving2021}
\begin{barticle}
\bauthor{\bsnm{Wang}, \binits{T.}},
\bauthor{\bsnm{Wu}, \binits{L.}},
\bauthor{\bsnm{Nobel}, \binits{P.}},
\bauthor{\bsnm{Roychowdhury}, \binits{J.}}:
\batitle{Solving combinatorial optimisation problems using oscillator based
  {{Ising}} machines}.
\bjtitle{Natural Computing}
(\byear{2021})
\doiurl{10.1007/s11047-021-09845-3}
\end{barticle}
\endbibitem

\bibitem[\protect\citeauthoryear{Wang and Roychowdhury}{2019}]{wangOIM2019}
\begin{bchapter}
\bauthor{\bsnm{Wang}, \binits{T.}},
\bauthor{\bsnm{Roychowdhury}, \binits{J.}}:
\bctitle{{{OIM}}: {{Oscillator-Based Ising Machines}} for {{Solving
  Combinatorial Optimisation Problems}}}.
In: \beditor{\bsnm{McQuillan}, \binits{I.}},
\beditor{\bsnm{Seki}, \binits{S.}} (eds.)
\bbtitle{Unconventional {{Computation}} and {{Natural Computation}}}
vol. \bseriesno{11493},
pp. \bfpage{232}--\blpage{256}.
\bpublisher{{Springer International Publishing}},
\blocation{{Cham}}
(\byear{2019}).
\doiurl{10.1007/978-3-030-19311-9_19}
\end{bchapter}
\endbibitem

\bibitem[\protect\citeauthoryear{Erementchouk
  et~al.}{2022}]{erementchoukComputational2022}
\begin{barticle}
\bauthor{\bsnm{Erementchouk}, \binits{M.}},
\bauthor{\bsnm{Shukla}, \binits{A.}},
\bauthor{\bsnm{Mazumder}, \binits{P.}}:
\batitle{On computational capabilities of {{Ising}} machines based on nonlinear
  oscillators}.
\bjtitle{Physica D: Nonlinear Phenomena}
\bvolume{437},
\bfpage{133334}
(\byear{2022})
\doiurl{10.1016/j.physd.2022.133334}
\end{barticle}
\endbibitem

\bibitem[\protect\citeauthoryear{B{\"o}hm
  et~al.}{2021}]{bohmOrderofmagnitude2021}
\begin{barticle}
\bauthor{\bsnm{B{\"o}hm}, \binits{F.}},
\bauthor{\bsnm{Vaerenbergh}, \binits{T.V.}},
\bauthor{\bsnm{Verschaffelt}, \binits{G.}},
\bauthor{\bsnm{{Van der Sande}}, \binits{G.}}:
\batitle{Order-of-magnitude differences in computational performance of analog
  {{Ising}} machines induced by the choice of nonlinearity}.
\bjtitle{Communications Physics}
\bvolume{4}(\bissue{1}),
\bfpage{149}
(\byear{2021})
\doiurl{10.1038/s42005-021-00655-8}
\end{barticle}
\endbibitem

\bibitem[\protect\citeauthoryear{Goemans and
  Williamson}{1994}]{goemans879approximation1994}
\begin{bchapter}
\bauthor{\bsnm{Goemans}, \binits{M.X.}},
\bauthor{\bsnm{Williamson}, \binits{D.P.}}:
\bctitle{.879-approximation algorithms for {{MAX CUT}} and {{MAX 2SAT}}}.
In: \bbtitle{Proceedings of the Twenty-Sixth Annual {{ACM}} Symposium on
  {{Theory}} of Computing - {{STOC}} '94},
pp. \bfpage{422}--\blpage{431}.
\bpublisher{{ACM Press}},
\blocation{{Montreal, Quebec, Canada}}
(\byear{1994}).
\doiurl{10.1145/195058.195216}
\end{bchapter}
\endbibitem

\bibitem[\protect\citeauthoryear{Goemans and
  Williamson}{1995}]{goemansImproved1995}
\begin{barticle}
\bauthor{\bsnm{Goemans}, \binits{M.X.}},
\bauthor{\bsnm{Williamson}, \binits{D.P.}}:
\batitle{Improved approximation algorithms for maximum cut and satisfiability
  problems using semidefinite programming}.
\bjtitle{Journal of the ACM}
\bvolume{42}(\bissue{6}),
\bfpage{1115}--\blpage{1145}
(\byear{1995})
\doiurl{10.1145/227683.227684}
\end{barticle}
\endbibitem

\bibitem[\protect\citeauthoryear{Burer et~al.}{2002}]{burerRankTwo2002}
\begin{barticle}
\bauthor{\bsnm{Burer}, \binits{S.}},
\bauthor{\bsnm{Monteiro}, \binits{R.D.C.}},
\bauthor{\bsnm{Zhang}, \binits{Y.}}:
\batitle{Rank-{{Two Relaxation Heuristics}} for {{MAX-CUT}} and {{Other Binary
  Quadratic Programs}}}.
\bjtitle{SIAM Journal on Optimization}
\bvolume{12}(\bissue{2}),
\bfpage{503}--\blpage{521}
(\byear{2002})
\doiurl{10.1137/S1052623400382467}
\end{barticle}
\endbibitem

\bibitem[\protect\citeauthoryear{Raghavendra}{2008}]{raghavendraOptimal2008}
\begin{bchapter}
\bauthor{\bsnm{Raghavendra}, \binits{P.}}:
\bctitle{Optimal algorithms and inapproximability results for every {{CSP}}?}
In: \bbtitle{Proceedings of the Fortieth Annual {{ACM}} Symposium on {{Theory}}
  of Computing}.
\bsertitle{{{STOC}} '08},
pp. \bfpage{245}--\blpage{254}.
\bpublisher{{Association for Computing Machinery}},
\blocation{{New York, NY, USA}}
(\byear{2008}).
\doiurl{10.1145/1374376.1374414}
\end{bchapter}
\endbibitem

\bibitem[\protect\citeauthoryear{Khot et~al.}{2004}]{khotOptimal2004}
\begin{bchapter}
\bauthor{\bsnm{Khot}, \binits{S.}},
\bauthor{\bsnm{Kindler}, \binits{G.}},
\bauthor{\bsnm{Mossel}, \binits{E.}},
\bauthor{\bsnm{O'Donnell}, \binits{R.}}:
\bctitle{Optimal {{Inapproximability Results}} for {{Max-Cut}} and {{Other}}
  2-{{Variable CSPs}}?}
In: \bbtitle{45th {{Annual IEEE Symposium}} on {{Foundations}} of {{Computer
  Science}}},
pp. \bfpage{146}--\blpage{154}.
\bpublisher{{IEEE}},
\blocation{{Rome, Italy}}
(\byear{2004}).
\doiurl{10.1109/FOCS.2004.49}
\end{bchapter}
\endbibitem

\bibitem[\protect\citeauthoryear{Burer and Monteiro}{2005}]{burerLocal2005}
\begin{barticle}
\bauthor{\bsnm{Burer}, \binits{S.}},
\bauthor{\bsnm{Monteiro}, \binits{R.D.C.}}:
\batitle{Local {{Minima}} and {{Convergence}} in {{Low-Rank Semidefinite
  Programming}}}.
\bjtitle{Mathematical Programming}
\bvolume{103}(\bissue{3}),
\bfpage{427}--\blpage{444}
(\byear{2005})
\doiurl{10.1007/s10107-004-0564-1}
\end{barticle}
\endbibitem

\bibitem[\protect\citeauthoryear{Burer and Monteiro}{2003}]{burerNonlinear2003}
\begin{barticle}
\bauthor{\bsnm{Burer}, \binits{S.}},
\bauthor{\bsnm{Monteiro}, \binits{R.D.C.}}:
\batitle{A nonlinear programming algorithm for solving semidefinite programs
  via low-rank factorization}.
\bjtitle{Mathematical Programming}
\bvolume{95}(\bissue{2}),
\bfpage{329}--\blpage{357}
(\byear{2003})
\doiurl{10.1007/s10107-002-0352-8}
\end{barticle}
\endbibitem

\bibitem[\protect\citeauthoryear{Boumal et~al.}{2016}]{boumalNonconvex2016}
\begin{bchapter}
\bauthor{\bsnm{Boumal}, \binits{N.}},
\bauthor{\bsnm{Voroninski}, \binits{V.}},
\bauthor{\bsnm{Bandeira}, \binits{A.S.}}:
\bctitle{The non-convex {{Burer}}\textendash{{Monteiro}} approach works on
  smooth semidefinite programs}.
In: \bbtitle{30 Th {{Conf}}. {{Neural Information Processing}} Systems
  ({{NIPS}} 2016)},
\bconflocation{{Barcelona, Spain}},
p. \bfpage{10}
(\byear{2016})
\end{bchapter}
\endbibitem

\bibitem[\protect\citeauthoryear{Boumal et~al.}{2020}]{boumalDeterministic2020}
\begin{barticle}
\bauthor{\bsnm{Boumal}, \binits{N.}},
\bauthor{\bsnm{Voroninski}, \binits{V.}},
\bauthor{\bsnm{Bandeira}, \binits{A.S.}}:
\batitle{Deterministic {{Guarantees}} for {{Burer-Monteiro Factorizations}} of
  {{Smooth Semidefinite Programs}}}.
\bjtitle{Communications on Pure and Applied Mathematics}
\bvolume{73}(\bissue{3}),
\bfpage{581}--\blpage{608}
(\byear{2020})
\doiurl{10.1002/cpa.21830}
\end{barticle}
\endbibitem

\bibitem[\protect\citeauthoryear{Bandeira et~al.}{2016}]{bandeiraLowrank2016}
\begin{bchapter}
\bauthor{\bsnm{Bandeira}, \binits{A.S.}},
\bauthor{\bsnm{Boumal}, \binits{N.}},
\bauthor{\bsnm{Voroninski}, \binits{V.}}:
\bctitle{On the low-rank approach for semidefinite programs arising in
  synchronization and community detection}.
In: \beditor{\bsnm{Feldman}, \binits{V.}},
\beditor{\bsnm{Rakhlin}, \binits{A.}},
\beditor{\bsnm{Shamir}, \binits{O.}} (eds.)
\bbtitle{29th Annual Conference on Learning Theory}.
\bsertitle{Proceedings of Machine Learning Research},
vol. \bseriesno{49},
pp. \bfpage{361}--\blpage{382}.
\bpublisher{{PMLR}},
\blocation{{Columbia University, New York, New York, USA}}
(\byear{2016})
\end{bchapter}
\endbibitem

\bibitem[\protect\citeauthoryear{Dunning et~al.}{2018}]{dunningWhat2018}
\begin{barticle}
\bauthor{\bsnm{Dunning}, \binits{I.}},
\bauthor{\bsnm{Gupta}, \binits{S.}},
\bauthor{\bsnm{Silberholz}, \binits{J.}}:
\batitle{What {{Works Best When}}? {{A Systematic Evaluation}} of
  {{Heuristics}} for {{Max-Cut}} and {{QUBO}}}.
\bjtitle{INFORMS Journal on Computing}
\bvolume{30}(\bissue{3}),
\bfpage{608}--\blpage{624}
(\byear{2018})
\doiurl{10.1287/ijoc.2017.0798}
\end{barticle}
\endbibitem

\bibitem[\protect\citeauthoryear{Shukla et~al.}{2022}]{shuklaScalable2022}
\begin{botherref}
\oauthor{\bsnm{Shukla}, \binits{A.}},
\oauthor{\bsnm{Erementchouk}, \binits{M.}},
\oauthor{\bsnm{Mazumder}, \binits{P.}}:
Scalable Almost-Linear Dynamical {{Ising}} Machines
(2022).
\doiurl{10.48550/arXiv.2205.14760}
\end{botherref}
\endbibitem

\bibitem[\protect\citeauthoryear{Shukla et~al.}{2023}]{shuklaCustom2023}
\begin{barticle}
\bauthor{\bsnm{Shukla}, \binits{A.}},
\bauthor{\bsnm{Erementchouk}, \binits{M.}},
\bauthor{\bsnm{Mazumder}, \binits{P.}}:
\batitle{Custom {{CMOS Ising Machine Based}} on {{Relaxed Burer-Monteiro-Zhang
  Heuristic}}}.
\bjtitle{IEEE Transactions on Computers}
\bvolume{72}(\bissue{10}),
\bfpage{2835}--\blpage{2846}
(\byear{2023})
\doiurl{10.1109/TC.2023.3272278}
\end{barticle}
\endbibitem

\bibitem[\protect\citeauthoryear{Punnen}{2022}]{punnenQuadratic2022}
\begin{bbook}
\beditor{\bsnm{Punnen}, \binits{A.P.}} (ed.):
\bbtitle{The Quadratic Unconstrained Binary Optimization Problem. {{Theory}},
  {{Algorithms}}, and Applications}.
\bpublisher{{Springer Nature Switzerland AG}},
\blocation{{Gewerbestrasse, Switzerland}}
(\byear{2022})
\end{bbook}
\endbibitem

\bibitem[\protect\citeauthoryear{Delacour
  et~al.}{2023}]{delacourMixedsignal2023}
\begin{barticle}
\bauthor{\bsnm{Delacour}, \binits{C.}},
\bauthor{\bsnm{Carapezzi}, \binits{S.}},
\bauthor{\bsnm{Boschetto}, \binits{G.}},
\bauthor{\bsnm{Abernot}, \binits{M.}},
\bauthor{\bsnm{Gil}, \binits{T.}},
\bauthor{\bsnm{Azemard}, \binits{N.}},
\bauthor{\bsnm{{Todri-Sanial}}, \binits{A.}}:
\batitle{A mixed-signal oscillatory neural network for scalable analog
  computations in phase domain}.
\bjtitle{Neuromorphic Computing and Engineering}
\bvolume{3}(\bissue{3}),
\bfpage{034004}
(\byear{2023})
\doiurl{10.1088/2634-4386/ace9f5}
\end{barticle}
\endbibitem

\bibitem[\protect\citeauthoryear{Filippov}{1988}]{filippovDifferential1988}
\begin{bbook}
\bauthor{\bsnm{Filippov}, \binits{A.F.}}:
\bbtitle{{Differential Equations with Discontinuous Righthand Sides}}.
\bpublisher{{Kluwer Academic Publishers}},
\blocation{{Dordrecht [Netherlands] ; Boston}}
(\byear{1988})
\end{bbook}
\endbibitem

\bibitem[\protect\citeauthoryear{Cortes}{2008}]{cortesDiscontinuous2008}
\begin{barticle}
\bauthor{\bsnm{Cortes}, \binits{J.}}:
\batitle{Discontinuous dynamical systems}.
\bjtitle{IEEE Control Systems Magazine}
\bvolume{28}(\bissue{3}),
\bfpage{36}--\blpage{73}
(\byear{2008})
\doiurl{10.1109/MCS.2008.919306}
\end{barticle}
\endbibitem

\bibitem[\protect\citeauthoryear{Steinerberger}{2023}]{steinerbergerMaxCut2023}
\begin{barticle}
\bauthor{\bsnm{Steinerberger}, \binits{S.}}:
\batitle{Max-{{Cut}} via {{Kuramoto-Type Oscillators}}}.
\bjtitle{SIAM Journal on Applied Dynamical Systems}
\bvolume{22}(\bissue{2}),
\bfpage{730}--\blpage{743}
(\byear{2023})
\doiurl{10.1137/21M1432211}
\end{barticle}
\endbibitem

\bibitem[\protect\citeauthoryear{Guglielmi and
  Hairer}{2022}]{guglielmiEfficient2022}
\begin{barticle}
\bauthor{\bsnm{Guglielmi}, \binits{N.}},
\bauthor{\bsnm{Hairer}, \binits{E.}}:
\batitle{An efficient algorithm for solving piecewise-smooth dynamical
  systems}.
\bjtitle{Numerical Algorithms}
\bvolume{89}(\bissue{3}),
\bfpage{1311}--\blpage{1334}
(\byear{2022})
\doiurl{10.1007/s11075-021-01154-1}
\end{barticle}
\endbibitem

\bibitem[\protect\citeauthoryear{Dembo et~al.}{2017}]{demboExtremal2017}
\begin{barticle}
\bauthor{\bsnm{Dembo}, \binits{A.}},
\bauthor{\bsnm{Montanari}, \binits{A.}},
\bauthor{\bsnm{Sen}, \binits{S.}}:
\batitle{Extremal cuts of sparse random graphs}.
\bjtitle{The Annals of Probability}
\bvolume{45}(\bissue{2}),
\bfpage{1190}--\blpage{1217}
(\byear{2017})
\doiurl{10.1214/15-AOP1084}
\end{barticle}
\endbibitem

\bibitem[\protect\citeauthoryear{Fan and Montanari}{2017}]{fanHow2017}
\begin{bchapter}
\bauthor{\bsnm{Fan}, \binits{Z.}},
\bauthor{\bsnm{Montanari}, \binits{A.}}:
\bctitle{How well do local algorithms solve semidefinite programs?}
In: \bbtitle{Proceedings of the 49th {{Annual ACM SIGACT Symposium}} on
  {{Theory}} of {{Computing}}}.
\bsertitle{{{STOC}} 2017},
pp. \bfpage{604}--\blpage{614}.
\bpublisher{Association for Computing Machinery},
\blocation{New York, NY, USA}
(\byear{2017}).
\doiurl{10.1145/3055399.3055451}
\end{bchapter}
\endbibitem

\bibitem[\protect\citeauthoryear{Yamamoto et~al.}{2017}]{yamamotoCoherent2017}
\begin{barticle}
\bauthor{\bsnm{Yamamoto}, \binits{Y.}},
\bauthor{\bsnm{Aihara}, \binits{K.}},
\bauthor{\bsnm{Leleu}, \binits{T.}},
\bauthor{\bsnm{Kawarabayashi}, \binits{K.-i.}},
\bauthor{\bsnm{Kako}, \binits{S.}},
\bauthor{\bsnm{Fejer}, \binits{M.}},
\bauthor{\bsnm{Inoue}, \binits{K.}},
\bauthor{\bsnm{Takesue}, \binits{H.}}:
\batitle{Coherent {{Ising}} machines---optical neural networks operating at the
  quantum limit}.
\bjtitle{npj Quantum Information}
\bvolume{3}(\bissue{1}),
\bfpage{49}
(\byear{2017})
\doiurl{10.1038/s41534-017-0048-9}
\end{barticle}
\endbibitem

\bibitem[\protect\citeauthoryear{Chen
  et~al.}{2022}]{Chen_cim-optimizer_a_simulator_2022}
\begin{botherref}
\oauthor{\bsnm{Chen}, \binits{F.}},
\oauthor{\bsnm{Isakov}, \binits{B.}},
\oauthor{\bsnm{King}, \binits{T.}},
\oauthor{\bsnm{Leleu}, \binits{T.}},
\oauthor{\bsnm{McMahon}, \binits{P.}},
\oauthor{\bsnm{Onodera}, \binits{T.}}:
{cim-optimizer: a Simulator of the Coherent Ising Machine}.
\url{https://github.com/mcmahon-lab/cim-optimizer}
\end{botherref}
\endbibitem

\end{thebibliography}


\end{document}